\documentclass[11pt]{llncs}

\usepackage{color}
\usepackage{url}

\usepackage[usenames,dvipsnames]{xcolor}
\usepackage{multirow}
\usepackage{graphics}
\usepackage{a4wide}
\usepackage{pdfpages}

\usepackage{todonotes}

\usepackage{caption}
\usepackage{graphicx}
\usepackage{booktabs}
\usepackage{enumitem}

\newcommand{\TBR}{{\rm TBR}}
\newcommand{\MP}{{\rm MP}}

\newcommand{\MAF}{{\rm MAF}}

\newcommand{\steven}[1]{\textcolor{black}{#1}}

\newcommand{\purple}{\textcolor{black}}

\newcommand{\stevenred}{\textcolor{black}}
\newcommand{\sidmablue}{\textcolor{black}}
\newcommand{\revision}[1]{\textcolor{black}{#1}}

\newtheorem{observation}{Observation}

\usepackage{amsmath}
\usepackage{graphicx}

\title{Deep kernelization for the Tree Bisection and Reconnnect (TBR) distance in phylogenetics}
\author{Steven Kelk\inst{1}, Simone Linz\inst{2}, Ruben Meuwese\inst{1} }

\institute{Department of Advanced Computing Science, Maastricht University, The Netherlands,\\ \email{steven.kelk@maastrichtuniversity.nl, ruben.meuwese@maastrichtuniversity.nl}
\and School of Computer Science, University of Auckland, New Zealand,\\
\email{s.linz@auckland.ac.nz}}

\providecommand{\keywords}[1]{\textit{Keywords:} #1}

\begin{document}
\maketitle

\begin{abstract}
We describe a kernel of size $9k-8$ for the NP-hard problem of computing the Tree Bisection and Reconnect (TBR) distance $k$ between two unrooted binary phylogenetic trees. \revision{To achieve this, we extend} the existing portfolio of reduction rules with three novel new reduction rules. Two of the rules are based on the idea of topologically transforming the trees in a distance-preserving way in order to guarantee execution of earlier reduction rules. The third rule extends the local neighbourhood approach introduced in \cite{kelk2020new} to more global structures, allowing new situations to be identified when deletion of a leaf definitely reduces the TBR distance by one. The bound on the kernel size is tight up to an additive term.
% leaves of the
%rees can safely be deleted.
Our results also apply to the equivalent problem of computing a Maximum Agreement Forest (MAF) between two unrooted binary phylogenetic trees. We anticipate that our results will be more widely applicable for computing agreement-forest based dissimilarity measures.
\end{abstract}


\section{Introduction}
\begin{figure}[t]
\center
\scalebox{1}{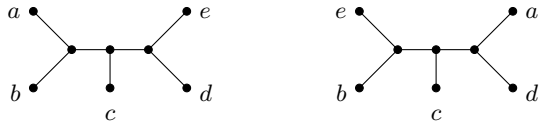}
\caption{Two unrooted binary phylogenetic trees on $X=\{a,b,c,d,e\}$.}
\label{fig:trees}
\end{figure}

A phylogenetic tree is essentially a tree in the usual graph-theoretical sense whose leaves are bijectively labeled by a set of labels $X$ \cite{SempleSteel2003}. Such trees have a central role in the study of evolution. \revision{The label set} $X$ represents a set of contemporary species, the \sidmablue{unlabeled interior vertices} of the tree represent hypothetical (extinct) ancestors of $X$ and the topology of the tree encodes the history of branching events, such as speciation, which caused those ancestors to diversify into the set of species $X$. A central challenge in the field of phylogenetics is to accurately infer such trees from data obtained solely from $X$, such as DNA data \cite{felsenstein2004inferring}. However, it is not uncommon to obtain
different trees for the same set $X$; this can be methodological (e.g. different objective functions or multiple optima) or due to the fact that some species have multiple distinct tree signals woven into their
genome \cite{som2015causes}.  This motivates the use of distance measures in phylogenetics, which rigorously quantify the dissimilarity of two phylogenetic trees. Such distance measures can communicate important information about
the biological significance of the observed differences \cite{Whidden2014} and can help us to understand the behaviour of tree-construction algorithms that traverse the space of phylogenetic trees by applying local rearrangement operations  \cite{john2017shape,money2012characterizing}. Distances can also be used as part of the toolkit for constructing non-treelike hypotheses of evolution, known as phylogenetic networks \cite{HusonRuppScornavacca10}.

In this article we are concerned with one such distance, Tree Bisection and Reconnect (TBR) distance, which is a metric on the space of unrooted (i.e. undirected) binary phylogenetic trees (see Figure \ref{fig:trees}). This distance represents the minimum number of times a subtree of one tree has to be detached, and reattached elsewhere, in order to transform it into the other tree (see Figure \ref{fig:tbr}). It is NP-hard to compute 
\cite{AllenSteel2001,hein1996complexity}. The problem has an equivalent, alternative formulation using \emph{agreement forests}. An agreement forest is a partition of $X$ such that the spanning trees induced by the blocks of the partition are disjoint in both trees \emph{and} the induced
spanning trees have the same topology in both trees, up to suppression of degree 2 vertices. An agreement forest with a minimum number of blocks is called a \emph{maximum agreement forest} (MAF); it is well-known that the TBR distance ($d_{\TBR}$) is equal to the number of blocks in a MAF ($d_{\MAF}$) minus 1 \cite{AllenSteel2001}. In the last \revision{twenty years} maximum agreement forests have received sustained attention from the mathematics, computer
science and bioinformatics communities, see e.g. \revision{\cite{AllenSteel2001,extremal2019,chen2015parameterized,downey2013fundamentals,kelk2016monadic,olver2022duality,whidden2013fixed}}. 
One response to the NP-hardness of computing $d_{\TBR}$ is \emph{kernelization}. Here the goal is to apply polynomial-time preprocessing rules such that $d_{\TBR}$ is preserved, or decreased in a controlled fashion,
such that the reduced trees  have at most $f(d_{\TBR})$ leaves for some function $f$ that depends only on $d_{\TBR}$. For further background on kernelization we refer to the book \cite{kernelization2019}. The core idea is that,
if $d_{\TBR}$ is small, then the reduced trees (known as the \emph{kernel}) will be small even if $|X|$ is very large \sidmablue{and $d_{\TBR}$} can be computed on these small trees using optimized exponential-time algorithms. The use of kernelization in this context is not coincidental: phylogenetics continues to be a rich source of open problems in, and application opportunities for, parameterized complexity \cite{bulteau2019parameterized}. \revision{Indeed, the applicability of techniques from parameterized algorithmics to problems in phylogenetics (e.g. agreement forests) have already been mentioned in \cite{downey2013fundamentals,Niedermeier:2002wh}.}

In 2001 it was shown in \cite{AllenSteel2001} that the \emph{subtree} and \emph{chain} reduction rules suffice to obtain a kernel of size at most $28k$, where $k$ is $d_{\TBR}$. These function by reducing common pendant subtrees and common chains (i.e. caterpillar-like regions), respectively. Almost 20 years later the present authors
proved that the same reduction rules actually yield a kernel of size at most $15k-9$, and in fact that this is tight \cite{tightkernel}. A critical insight in \cite{tightkernel} was that computation of $d_{\TBR}$ (or  $d_{\MAF}$) can equivalently be viewed as the problem of adding the labels $X$, and a set of \emph{breakpoints} (essentially: edge cuts), to an (unknown) cubic multigraph, known as a \emph{generator}, such that the original two trees can be retrieved \stevenred{(see Figure \ref{fig:breakpoints})}. This insight was subsequently leveraged in \cite{kelk2020new} to design five new reduction rules which, when added to the subtree and chain reduction rules,
yield a tight kernel of size $11k-9$. An empirical follow-up showed that the new rules in \cite{kelk2020new} have added reductive power in practice \cite{van2022reflections}, and recently similar techniques have been used to design new
reduction rules for distances and agreement forests on \emph{rooted} trees \cite{kelk2022cyclic}.

\purple{Following the reduction in the size of the $d_{\TBR}$ kernel from $28k$ to $15k-9$, and then to $11k-9$,  it is natural to ask: can we do better than $11k-9$? In this article we answer the question affirmatively: we give a kernel of size $9k-8$, which is tight up to an additive term. We note that such an ongoing research effort is certainly not unprecedented in the parameterized complexity literature. For example, in a sequence of articles the kernel for the (unrelated) Feedback Vertex Set problem on planar graphs was progressively reduced from $112k$ to $13k$,  where $k$ is the size of a feedback vertex set of the input \cite{longabu2012,bodlaender2008long,bonamy201613k,xiao2014long}. Our result fits in this tradition.}

%For an overview of smallest-known kernels for a range of algorithmic problems, we refer the interested reader to \cite{fpt-wiki}

\purple{To obtain a kernel of size $9k-8$} we use the analytical and counting bottlenecks identified in \cite{kelk2020new} as a starting point, and use these to guide the design of three \revision{new}
%, novel 
reduction rules. The reduction rules have a \stevenred{very}
%rather
different \sidmablue{flavor} to what has come before. The
first new reduction rule addresses the following bottleneck: some of the topological structures that contribute heavily to the $11k-9$ bound, and which we thus wish to target for reduction, could potentially be leveraged
by a depth-bounded branching algorithm that recursively cuts edges in the input trees to obtain an agreement forest. However, the cuts applied by such a direct, non-preprocessing algorithm yield a different, more general problem,
on forests rather than trees, which is analytically far harder to deal with from a kernelization perspective. The first new reduction rule, Reduction 8, circumvents this by applying a $d_{\TBR}$-preserving transformation to one of the
trees, such that the classical subtree reduction rule can be applied and the number of leaves can be reduced; in this way we stay in the world of trees. The transformation itself requires a very careful analysis of the way common chains behave when
one of the chains is `interrupted' in the other tree. Essentially, the transformation works by deleting an edge in one of the trees and replacing it with an edge that is  `buried' inside an artificially lengthened common chain, which ensures that $d_{\TBR}$ does
not change.  \stevenred{Notably, the artificially lengthened chain is obtained by reversing the classical chain reduction rule\sidmablue{.} Reduction 8 is thus an example of where newer reduction rules make progress by undoing earlier reduction rules (see \cite{figiel2022there} for related discussions).}

The second new reduction rule, Reduction 9, works by identifying other topological structures which contribute heavily to the $11k-9$ bound, and transforming them into structures that
can be attacked by Reduction 8. Reduction 9 only applies when the region surrounding the topological structure contains many leaves\revision{. Conversely,} if Reduction 9 does not apply, the region is sparse. Reduction
10 is similar in spirit to Reduction 9, but is more direct: if it triggers, it is parameter reducing i.e. $d_{\TBR}$ is definitely reduced by 1. Once Reductions 8--10 no longer apply (or the earlier reduction rules),
there is extensive sparsity in the underlying generator, which we use to obtain the new bound of $9k-8$. We show that this bound is (essentially) tight by describing irreducible pairs of trees with TBR distance
$k$ that have $9k-9$ leaves.

We anticipate that the new reduction rules will yield new advances for other agreement-forest based distances in phylogenetics, contribute to a deeper understanding of the combinatorics of agreement forests, and facilitate the ongoing advancement of kernelization within phylogenetics. 

\section{Preliminaries}

\subsection{Notation and terminology}
Our notation closely follows \cite{kelk2020new}. Throughout this paper, $X$ denotes a non-empty finite set of \emph{taxa}.\\
% with $|X| \geq 4$. 

\noindent{\bf Phylogenetic trees.} An {\it unrooted binary phylogenetic tree} $T$ on $X$  is a  simple, connected, and undirected tree whose leaves are bijectively labeled with $X$ and whose other vertices all have degree 3. The set $X$ is often referred to as the {\it leaf set} of $T$. See Figure \ref{fig:trees} for an example of two unrooted binary phylogenetic trees on $X=\{a,b,c,d,e\}$. For simplicity and since \revision{all} phylogenetic trees in this paper are unrooted and binary, we refer to an unrooted binary phylogenetic trees as a {\it phylogenetic tree}.
%If a definition or statement applies to all unrooted phylogenetic trees, regardless of whether they are binary or not, we make this explicit. 
Two leaves, say $a$ and $b$, of $T$ are called a {\it cherry} $\{a,b\}$ of $T$ if they are adjacent to a common vertex. Moreover, for each $x\in X$, we use $p_x$ to denote the unique neighbor of $x$ in $T$ and refer to $p_x$ as the {\it parent} of $x$.
%Let $T$ be an unrooted binary phylogenetic tree on $X$.
%A {\it quartet} is an unrooted binary phylogenetic tree with exactly four leaves. If $\{a,b,c,d\}\subseteq X$, we say that $ab|cd$ is a quartet of $T$ if  the path from $a$ to $b$ does not intersect the path from $c$ to $d$. Note that, if $ab|cd$ is not a quartet of $T$, then either $ac|bd$ or $ad|bc$ is a quartet of $T$.  \\
%We say that a vertex $v$ is the (unique) {\it parent} of a leaf $a$ in $N$ if $v$ is adjacent to  $a$. 

For $X'  \subseteq X$, we write $T[X']$ to denote the unique, minimal subtree of $T$ that connects all elements in $X'$. For brevity we call $T[X']$ the \emph{embedding} of  $X'$ in $T$. For an edge $e$ of $T$, we say that $T[X']$ {\it uses} $e$, if $e$ is an edge of $T[X']$.
%If $X''$ is also a subset of $X$, we denote by $T[X']\cap T[X'']$ the set of \textcolor{black}{vertices} in $T$ that are contained in $T[X']$ and $T[X'']$. 
Furthermore, we refer to the phylogenetic tree on $X'$ obtained from $T[X']$ by suppressing degree-2 vertices as  the {\it restriction of $T$ to $X'$}  and we denote this by $T|X'$. \\

\noindent{\bf Subtrees and chains.}
Let $T$ be a phylogenetic tree on $X$.  We say that a subtree of $T$ is {\it pendant} if it can be detached from $T$ by deleting a single edge. For $n\geq 2$, let $C = (\ell_1,\ell_2\ldots,\ell_n)$ be a sequence of distinct taxa in $X$. 
%For each $i\in\{1,2\ldots,n\}$, let $p_i$ denote the unique parent of $\ell_i$ in $T$. 
We call $C$ an $n$-chain of $T$ if there exists a walk \sidmablue{$p_{\ell_1},p_{\ell_2},\ldots,p_{\ell_n}$} in $T$ and the elements in \sidmablue{$p_{\ell_2},p_{\ell_3},\ldots,p_{\ell_{n-1}}$} are all pairwise distinct. Note that $\ell_1$ and $\ell_2$ may have a common parent or $\ell_{n-1}$ and $\ell_n$ may have a common parent. Furthermore, if  \sidmablue{$p_{\ell_1} = p_{\ell_2}$ or $p_{\ell_{n-1}} = p_{\ell_n}$} holds, then $C$ is said to be {\it pendant} in $T$. To ease reading, we sometimes write $C$ to denote the set $\{\ell_1,\ell_2,\ldots,\ell_n\}$. It will always be clear from the context whether $C$ refers to the associated sequence or set of taxa. If a pendant subtree $S$ (resp. an $n$-chain $C$) exists in two phylogenetic trees $T$ and $T'$ on $X$, we say that $S$ (resp. $C$) is a {\it common} subtree (resp. chain) of $T$ and $T'$. \\
%Moreover, if $C$ is a common $n$-chain of $T$ and $T'$, reducing $C$ to a chain of length $k$ with $1\leq k <n$ yields the two new trees $T_r = T|X\setminus\{\ell_{k+1},\ell_{k+2},\ldots,\ell_n\}$ and $T_r' = T'|X\setminus\{\ell_{k+1},\ell_{k+2},\ldots,\ell_n\}$.\\

\noindent{\bf Tree bisection and reconnection.} Let $T$ be a phylogenetic tree on $X$. Apply the following three-step operation to $T$:
\begin{enumerate}
\item Delete an edge in $T$ and suppress any resulting degree-2 vertex. Let $T_1$ and $T_2$ be the two resulting phylogenetic trees.
\item If $T_1$ (resp. $T_2$) has at least one edge, subdivide an edge in $T_1$ (resp. $T_2$) with a new vertex $v_1$ (resp. $v_2$) and otherwise set $v_1$ (resp. $v_2$) to be the single isolated vertex of $T_1$ (resp. $T_2$).
\item Add a new edge $\{v_1,v_2\}$ to obtain a new phylogenetic tree $T'$ on $X$.
% it was (v1,v2), I changed it to {v1, v2}
\end{enumerate}
We say that $T'$ has been obtained from $T$ by a single  {\it tree bisection and reconnection (TBR) operation}.
Furthermore, we define the TBR {\it distance}  between two phylogenetic trees $T$ and $T'$ on $X$,  denoted by $d_\TBR(T,T')$, to be the minimum number of TBR operations that are required to transform $T$ into $T'$. 
To illustrate, the trees $T$ and $T'$ in Figure \ref{fig:tbr} have a TBR distance of 1.
It is well known that $d_\TBR$ is a metric~\cite{AllenSteel2001}. By building on an earlier result by Hein et al.~\cite[Theorem 8]{hein1996complexity}, Allen and Steel~\cite{AllenSteel2001} showed that computing the TBR distance is an NP-hard problem. \\

\begin{figure}[t]
\center
\scalebox{1}{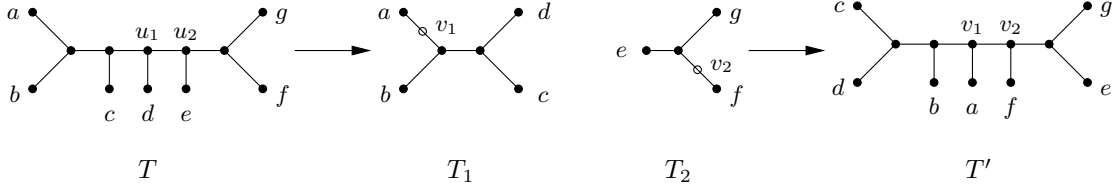}
\caption{A single TBR operation that transforms $T$ into $T'$. First, \sidmablue{up to ignoring the open circle degree-2 vertices,} $T_1$ and $T_2$ are obtained from $T$ by deleting the edge $\{u_1,u_2\}$ in $T$. Second, $T'$ is obtained from $T_1$ and $T_2$ by subdividing an edge in both trees as indicated by the open circles $v_1$ and $v_2$ and adding a new edge $\{v_1,v_2\}$.}
\label{fig:tbr}
\end{figure}

\noindent {\bf Agreement forests.}
Let $T$ and $T'$ be two phylogenetic trees  on $X$. Furthermore, let $F = \{B_0, B_1,B_2,\ldots,B_k\}$ be a partition of $X$, where each block $B_i$ with $i\in\{0,1,2,\ldots,k\}$ is  referred to as a \emph{component} of $F$. We say that $F$ is an \emph{agreement forest} for $T$ and $T'$ if the following conditions hold. 
\begin{enumerate}
\item [(1)] For each $i\in\{0,1,2,\ldots,k\}$, we have $T|B_i = T'|B_i$.
\item [(2)] For each pair $i,j\in\{0,1,2,\ldots,k\}$ with $i \neq j$, we have that
$T[B_i]$ and $T[B_j]$ are vertex-disjoint in $T$, and $T'[B_i]$ and $T'[B_j]$ are vertex-disjoint in $T'$. 
\end{enumerate}
\noindent
Let $F=\{B_0,B_1,B_2,\ldots,B_k\}$ be an agreement forest for $T$ and $T'$. The \emph{size} of $F$ is simply its number of components; i.e. $k+1$. Moreover, an agreement forest with the minimum number of components (over all agreement forests for $T$ and $T'$) is called a \emph{maximum agreement forest (MAF)} for $T$ and $T'$. The number of components of a maximum agreement forest for $T$ and $T'$ is denoted by $d_\MAF(T,T')$. The following theorem is well known.

\begin{theorem}\cite[Theorem 2.13]{AllenSteel2001}
Let $T$ and $T'$ be two phylogenetic trees  on $X$. Then $$d_\TBR(T,T') = d_\MAF(T,T') - 1.$$
\end{theorem}

\steven{A maximum agreement forest for the trees $T$ and $T'$ shown in Figure \ref{fig:tbr}, which have TBR distance 1, therefore contains two components. \revision{Here} $F=\{\{a,b,c,d\},\{e,f,g\}\}$ \revision{is the unique maximum agreement forest for $T$ and $T'$.}}

%An algorithm is said to be \emph{fixed parameter tractable} (FPT) with respect to a parameter $k$ if it has running time $O( f(k) \cdot \text{poly}(n) )$ where $n$ is the size of
%the input and $f$ is a computable function that depends ony on $k$. We refer to a standard text such as \cite{Cygan:2015:PA:2815661} for more background on FPT algorithms. In this article we take $d_{\TBR}$ as the parameter
%$k$ and $n=|X|$.  The concept of \emph{kernelization} is closely related. We refer 

\subsubsection{Phylogenetic networks.}
An {\it unrooted binary phylogenetic network} $N$ on $X$  is a  simple, connected, and undirected graph whose leaves are bijectively labeled with $X$ and whose other vertices all have degree 3.  Let $E$ and $V$ be the edge and vertex set of $N$, respectively. As with phylogenetic trees, we refer to an unrooted binary phylogenetic network simply as a {\it phylogenetic network}. Furthermore, we define the {\it reticulation number} of a phylogenetic network $N$ as the number of edges in $E$ that need to be deleted from $N$ to obtain a spanning tree. More formally, we have $r(N) = |E|-(|V|-1)$. If $r(N)=0$, then $N$ is simply a phylogenetic tree on $X$.

%Let $T$ be an unrooted binary phylogenetic tree on $X$. A {\it quartet} is an unrooted binary phylogenetic tree with exactly four leaves. For example, if $\{a,b,c,d\}\subseteq X$ , we say that $ab|cd$ is a quartet of $T$ if  the path from $a$ to $b$ does not intersect the path from $c$ to $d$. Note that, if $ab|cd$ is not a quartet of $T$, then either $ac|bd$ or $ad|bc$ is a quartet of $T$.  \\

Let $N$ be a phylogenetic network on $X$, and let $T$ be a phylogenetic tree on $X$. 
We say that $N$ \sidmablue{\it displays} $T$ if, up to suppressing degree-two vertices, $T$ can be obtained from $N$ by deleting edges and vertices, in which case, the resulting subgraph of $N$ is an {\it image} 
of $T$ in $N$. Observe that an image of $T$ in $N$ is a subdivision of $T$. See \stevenred{Figure \ref{fig:breakpoints} for an example \sidmablue{of the notion of displaying}.}

\subsubsection{Generators.}

Let $k$ be a positive integer.
For $k\geq 2$, a {\it $k$-generator} (or short {\it generator} when $k$ is clear from the context) is a connected cubic multigraph with edge set $E$ and vertex set $V$ such that $k=|E|-(|V|-1)$. 
%Furthermore, for $k=1$, we define the graph that consists of a single vertex $u$ and a loop edge $\{u,u\}$ to be the unique {\it $1$-generator}. 
The edges of a generator are  called its {\it sides}. Intuitively, given a phylogenetic network $N$ with $r(N)=k$,
%and no pendant subtree with at least two leaves, 
we can obtain a $k$-generator by, repeatedly, deleting all (labeled and unlabeled) leaves and suppressing any resulting degree-2 vertices. We say that the generator obtained in this way {\it underlies} $N$. \sidmablue{An example of a 2-generator is shown in Figure \ref{fig:breakpoints}.} Now, let $G$ be a $k$-generator, let $\{u,v\}$ be a side of $G$, and let $Y$ be a set of leaves. The operation of subdividing $\{u,v\}$ with $|Y|$ new vertices and, for each such new vertex $w$, adding a new edge $\{w,\ell\}$, where $\ell\in Y$ and $Y$ bijectively labels the new leaves, is referred to as {\it attaching} $Y$  to  $\{u,v\}$ or as {\it decorating $\{u,v\}$ with $Y$}. 
%Additionally, if $G$ is the 1-generator, then the degree-2 vertex $u$ is suppressed after attaching $Y$ to $\{u,u\}$. 
Lastly, if at least one new leaf is attached to each loop and to each pair of parallel edges in $G$, then the resulting graph is a phylogenetic network $N$ with $r(N)=k$. Note that $N$ has no pendant subtree with more than a single leaf.

Hence, we have the following observation.

\begin{observation}\label{ob:gen}
Let $N$ be a phylogenetic network that has no pendant subtree with at least two leaves, and let $G$ be a generator. 
Then $G$ underlies $N$ if and only if $N$ can be obtained from $G$ by attaching a (possibly empty) set of leaves to each side of $G$.
\end{observation}

\subsubsection{Unrooted minimum hybridization.} In \cite{van2018unrooted}, it was shown that computing the TBR distance for a pair of phylogenetic trees $T$ and $T'$ on $X$ is equivalent to computing the minimum number of extra edges required to simultaneously explain $T$ and $T'$. More precisely, we set
$$h^u(T, T') = \min_ N\{r(N)\},$$ where the minimum is taken over all phylogenetic networks $N$ on $X$ that display $T$ and $T'$ \stevenred{(and possibly other \sidmablue{phylogenetic} trees)}. The value $h^u(T, T')$ is known as the {\it (unrooted) hybridization number} of $T$ and $T'$ \cite{van2018unrooted}.\\

\noindent The aforementioned equivalence is given in the next theorem that was established in~\cite[Theorem 3]{van2018unrooted}.

\begin{theorem}\label{t:tbr-equiv}
Let $T$ and  $T'$ be two phylogenetic trees on $X$. Then $$d_\TBR(T,T')=h^u(T,T').$$
\end{theorem}

This means that $d_\TBR(T,T') = k$ if and only if there exists a phylogenetic network $N$ with $r(N)=k$ that displays both $T$ and $T'$. Such an $N$ can be obtained from its underlying generator, which has exactly $3(k-1)$ sides~\cite[Lemma 1]{tightkernel},
by attaching taxa to sides. The articles \cite{tightkernel,kelk2020new} use this fact
extensively to derive a bound on the size of the kernelized instance. We will use the same generator-based framework for our results. \\

\begin{figure}[t]
\center
\scalebox{0.9}{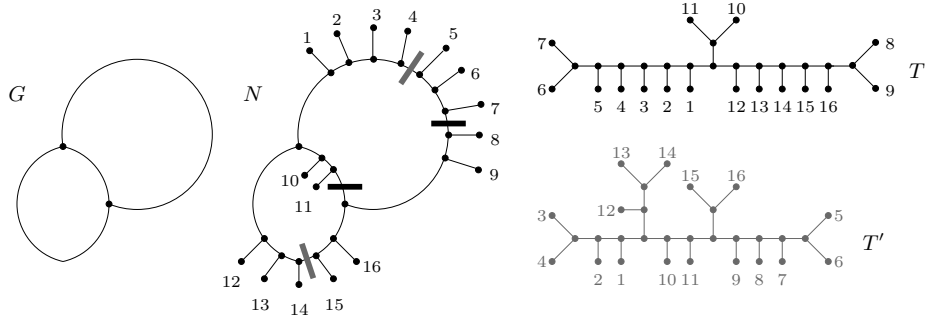}
\caption{\sidmablue{The generator $G$ (left) underlying the phylogenetic network $N$ (middle) that displays $T$ and $T'$ (right)}. An \revision{image} of $T$ (respectively, $T'$) can be obtained by deleting the $r(N)=d_{\TBR}(T,T')=2$ black (respectively, \sidmablue{gray}) breakpoints. The generator underlying $N$ has three sides: a 2-breakpoint side with 9 taxa, $1234|567|89$, and
two 1-breakpoint sides. Note that due to the common chain $C=(1,2,3,4)$ these trees could be reduced further by Reduction 2.}
\label{fig:breakpoints}
\end{figure}

\noindent{\bf Parameterized algorithms.} A \emph{parameterized problem} is a problem for which the inputs are of the form $(x,k)$, where $k$ is a non-negative integer, called the \emph{parameter}. A parameterized problem is \emph{fixed-parameter tractable} (FPT) if there exists an algorithm that solves\footnote{Note that the formalism described here actually concerns \emph{decision} (i.e. yes/no) problems, which in the context of the current article is most naturally ``Is $d_{\TBR}(T,T') \leq k$?''. An FPT algorithm for answering this question can easily be transformed into an algorithm for computing $d_{\TBR}$ with similar asymptotic time complexity by increasing $k$ incrementally from 0 until a yes-answer is obtained.} any instance $(x,k)$ in $f(k)\cdot |x|^{O(1)}$ time, where $f(\cdot)$ is a computable function depending only on $k$. A parameterized problem has a \emph{kernel} of size $g(k)$, where $g(\cdot)$ is a computable function depending only on $k$, if there exists a polynomial time algorithm transforming any instance $(x,k)$ into an equivalent problem $(x',k')$, with $|x'|,k' \leq g(k)$. Informally, this polynomial-time algorithm usually consists of reduction rules that are applied to an instance $(x,k)$ to transform it into an equivalent but smaller instance $(x',k')$.
If $g(k)$ is a polynomial in $k$ then we call this a \emph{polynomial kernel}; if $g(k) = O(k)$ then it is a \emph{linear kernel}. It is well-known that a parameterized problem is fixed-parameter tractable if and only if it has a (not necessarily polynomial) kernel. For more background information on fixed parameter tractability and kernelization, we refer the reader to standard texts such as~\cite{Cygan:2015:PA:2815661,downey2013fundamentals,kernelization2019}. 

Let $T$ and $T'$ be two phylogenetic trees on $X$. To compute $d_\TBR(T,T')$, we take $d_{\TBR}$ as the parameter $k$ and take $|X|$, the number of leaves, as the size of the instance $|x|$. The reduction rules described in the following section produce a linear kernel and run in $\text{poly}(|X|)$ time.

\subsection{Seven reductions to kernelize the TBR distance}
\label{subsc:def_reducrules}

We start this section by describing the existing seven reductions that have previously been used to establish kernelization results for computing the TBR distance. These existing reductions will be extended to ten reductions in Section~\ref{sec:new-rules}. \\

Let $T$ and $T'$ be two phylogenetic trees on $X$. The seven reductions are as follows.\\

\noindent {\bf Reduction 1.}~\cite{AllenSteel2001} If $T$ and $T'$ have a maximal common pendant subtree $S$ with at least two leaves, then reduce $T$ and $T'$ to $T_r$ and $T'_r$, respectively, by replacing $S$ with a single leaf with a new label.

\noindent {\bf Reduction 2.}~\cite{AllenSteel2001} If $T$ and $T'$ have a maximal common $n$-chain $C=(\ell_1,\ell_2,\ldots,\ell_n)$ with $n\geq 4$, then reduce $T$ and $T'$ to $T_r=T|X\setminus \{\ell_4,\ell_5,\ldots,\ell_n\}$ and $T_r'=T'|X\setminus \{\ell_4,\ell_5,\ldots,\ell_n\}$, respectively.

\noindent {\bf Reduction 3.}~\cite{kelk2020new} If $T$ and $T'$ have a common 3-chain $C=(\ell_1,\ell_2,\ell_3)$ such that $\{\ell_1,\ell_2\}$ is a cherry in $T$ and $\{\ell_2,\ell_3\}$ is a cherry in $T'$, then reduce $T$ and $T'$ to $T_r=T|X\setminus C$ and $T_r'=T'|X\setminus C$, respectively.

\noindent {\bf Reduction 4.}~\cite{kelk2020new} If $T$ and $T'$ have a common 3-chain $C=(\ell_1,\ell_2,\ell_3)$ such that $\{\ell_2,\ell_3\}$ is a cherry in $T$ and $\{\ell_3,x\}$ is a cherry in $T'$ with $x\in X\setminus C$, then reduce $T$ and $T'$ to $T_r=T|X\setminus \{x\}$ and $T_r'=T'|X\setminus \{x\}$, respectively.

\noindent {\bf Reduction 5.}~\cite{kelk2020new} If $T$ and $T'$ have two common 2-chains $C_1=(\ell_1,\ell_2)$ and $C_2=(\ell_3,\ell_4)$ such that $T$ has cherries $\{\ell_2,x\}$ and $\{\ell_3,\ell_4\}$, and $T'$ has cherries $\{\ell_1,\ell_2\}$ and $\{\ell_4,x\}$ with $x\in X\setminus (C_1\cup C_2)$, then reduce $T$ and $T'$ to $T_r=T|X\setminus \{x\}$ and $T_r'=T'|X\setminus \{x\}$, respectively.

\noindent {\bf Reduction 6.}~\cite{kelk2020new} If $T$ and $T'$ have two common 3-chains $C_1=(\ell_1,\ell_2,\ell_3)$ and $C_2=(\ell_4,\ell_5,\ell_6)$ such that $T$ has cherries $\{\ell_2,\ell_3\}$ and $\{\ell_4,\ell_5\}$, and $(\ell_1,\ell_2,\ldots,\ell_6)$ is a 6-chain of $T'$, then reduce $T$ and $T'$ to $T_r=T|X\setminus \{\ell_4,\ell_5\}$ and $T_r'=T'|X\setminus \{\ell_4,\ell_5\}$, respectively.

\noindent {\bf Reduction 7.}~\cite{kelk2020new} If $T$ and $T'$ have common chains $C_1=(\ell_1,\ell_2,\ell_3)$ and $C_2=(\ell_4,\ell_5)$ such that $T$ has cherries $\{\ell_2,\ell_3\}$ and $\{\ell_4,\ell_5\}$, and $(\ell_1,\ell_2,\ldots,\ell_5)$ is a 5-chain of $T'$, then reduce $T$ and $T'$ to $T_r=T|X\setminus \{\ell_4\}$ and $T_r'=T'|X\setminus \{\ell_4\}$, respectively.\\

\noindent An example of Reduction 7 is illustrated in Figure~\ref{fig:reduction7}.

Reduction 1 is known as {\it subtree reduction} while Reduction 2 is known as {\it chain reduction} in the literature. Now, suppose that two phylogenetic trees $T_r$ and $T_r'$ have a common 3-chain  $C=(\ell_1,\ell_2,\ell_2)$. We refer to the reverse of Reduction 2 which is the process of obtaining $T$ and $T'$ from $T_r$ and $T_r'$, respectively, as {\it extending $C$ to an $n$-chain} for $n>3$. We will always explicitly say in which order and
% on
to which end of $C$ we add the new leaves $\ell_4,\ell_5,\ldots,\ell_n$.

The following  lemma and theorem summarize results established in~\cite{AllenSteel2001,tightkernel,kelk2020new}.

\begin{figure}[t]
\center
\scalebox{1}{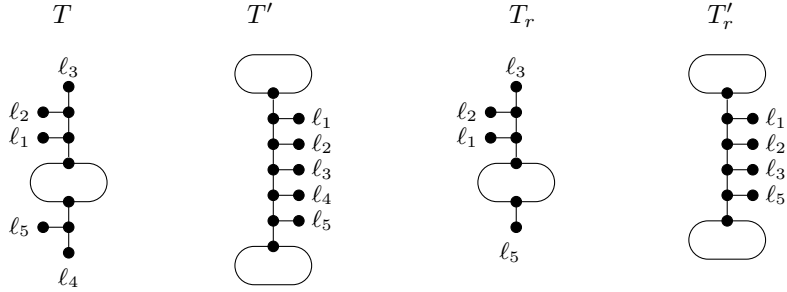}
\caption{An example of Reduction 7. Ovals indicate subtrees. }
\label{fig:reduction7}
\end{figure}

\begin{lemma}\label{l:7-reductions}
Let $T$ and $T'$ be two phylogenetic trees on $X$.
%with $|X|\geq 4$. 
If $T_r$ and $T_r'$ are two phylogenetic trees obtained from $T$ and $T'$, respectively, by  a single application of Reduction 1,2, 6, or 7, then $d_\TBR(T,T') = d_\TBR(T_r,T_r')$. Moreover, if $T_r$ and $T_r'$ are two trees obtained from $T$ and $T'$, respectively, by a single application of Reduction 3, 4, or 5, then $d_\TBR(T,T')-1 = d_\TBR(T_r,T_r')$.
\end{lemma}

\begin{theorem}\label{t:old-kernel}
Let $S$ and $S'$ be two phylogenetic trees on $X$  that cannot be reduced by Reduction 1 or 2, and let $T$ and $T'$ be two phylogenetic trees on $Y$ that cannot be reduced by any of Reductions 1--7. If $d_\TBR(S,S')\geq 2$, then $|X|\leq 15d_\TBR(S,S')-9$. Furthermore, if $d_\TBR(T,T')\geq 2$, then $|Y|\leq 11d_\TBR(T,T')-9$.
\end{theorem}

Note that each of Reductions 3, 4, and 5 triggers a \emph{parameter reduction}, whereby the TBR distance is reduced by one. In these cases, an element of $X$ is located which definitely comprises a singleton component in some maximum agreement forest, and whose deletion thus lowers the TBR distance by 1. Reductions 1, 2, 6 and 7, on the other hand, preserve the TBR distance. Reduction 6 and 7 work by truncating short chains, i.e. chains which escape Reduction 2, to be even shorter.

The following minor observation is worth noting.

\begin{observation}
\label{obs:pendancy}
Assume that Reductions 1--7 have been applied to exhaustion. 
Suppose $T$ and $T'$ have a common chain $C = (b,c,d)$ that is pendant in $T'$. Then $C$ is not pendant in $T$.
\end{observation}
\begin{proof}
If $C$ was pendant in $T$ then at least one of the subtree reduction or Reduction 3 would be applicable on $C$, contradicting the assumption that the reduction rules had been applied to exhaustion. $\qed$
\end{proof}

We end this section by outlining some of the machinery used in~\cite{tightkernel,kelk2020new} to kernelize the TBR distance. This article builds on that machinery and further refines it. Let $T$ and $T'$ be two phylogenetic trees on $X$ that cannot be reduced under Reduction 1 or 2,
and let $N$ be a phylogenetic network on $X$ that displays $T$ and $T'$.  Let $R$ and $R'$ be spanning trees of $N$ obtained by greedily extending an \revision{image} of $T$ (respectively, $T'$) to become a spanning tree,
if it is not that already.
Since $N$ displays $T$ and $T'$, $R$ and $R'$ exist. Furthermore, let $G$ be the generator that underlies $N$. Since $T$ and $T'$ are subtree and chain reduced, $N$ does not have a pendant subtree of size at least two. Hence, by Observation~\ref{ob:gen}, we can obtain $N$ from $G$ by attaching leaves to $G$. Let $S=\{u,w\}$ be a side of $G$. Let $Y=\{\ell_1,\ell_2,\ldots,\ell_m\}$ be the set of leaves that are  attached to $S$ in obtaining $N$ from $G$. Recall that $m \geq 0$. Then there exists a path $$u=v_0,v_1,v_2,\ldots,v_m,v_{m+1}=w$$ of vertices in $N$ such that, for each $i\in\{1,2,\ldots,m\}$, $v_i$ is the unique neighbor of $\ell_i$. We refer to this path as the {\it path associated with $S$} and denote it by $P_S$. Importantly, for a path $P_S$ in $N$ that is associated with a side $S$ of $G$, there is at most one edge in $P_S$ that is not contained in $R$, and there is at most one (not necessarily distinct) edge in $P_S$ that is not contained in $R'$. We  make this precise in the following definition and say that $S$ {\it is a $b$-breakpoint side} relative to $R$ and $R'$, 
where 
\begin{enumerate}
\item $b=0$ if $R$ and $R'$ both contain all edges of $P_S$,
\item $b=1$ if one element in $\{R,R'\}$ contains all edges of $P_S$ while the other element contains all but one edge of $P_S$, and
\item $b=2$ if each of $R$ and $R'$ contains all but one edge of $P_S$.
\end{enumerate}
Since $R$ and $R'$ span $N$, note that $S$ cannot have more than two breakpoints relative to $R$ and $R'$.  Let $S=\{u,w\}$ be a side of $G$ to which four taxa get attached in obtaining $N$ from $G$, and let $P_S=u,p_a,p_b, p_c,p_d,w$ be the path associated with $S$. For shorthand we will throughout this article use notation such as $2|2$ or $S=ab|cd$ to refer to a side $S$ if $P_S$ has a single breakpoint such that one of $R$ and $R'$ does not contain the edge $\{p_b,p_c\}$, and $2|1|1$ or $S=ab|c|d$ to refer to a side $S$ if $P_S$ has two breakpoints such that one of $R$ and $R'$ does not contain the edge $\{p_b,p_c\}$ and the other does not contain the edge $\{p_c,p_d\}$. \stevenred{(See Figure \ref{fig:breakpoints} for an example illustrating breakpoint notation).}
If $R$ and $R'$ both have the same breakpoint (i.e. there exists an edge of $P_S$ that neither $R$ nor $R'$ contains), then we write, for example, $2||2$ or $1||3$. Lastly, note that there also may exist a side such that $R$ or $R'$ does not contain the edge $\{u,p_a\}$ or $\{p_d,w\}$ in which case we write, for example, $0|2|2$, $0|4|0$, or $0|4$. Similar notation extends to sides in $G$ to which three taxa get attached in obtaining $N$ from $G$.

\section{Two technical results about short chains}

This section present two technical but powerful theorems that play a crucial part in the upcoming sections. The first, Theorem~\ref{thm:allchainsintact}, was established in~\cite[Theorem 5]{kelk2020new}, while the second, Theorem~\ref{t:lovely-chains}, is new to this paper. \\

Let $F=\{B_0,B_1,B_2,\ldots,B_k\}$ be an agreement forest for two phylogenetic trees $T$ and $T'$ on $X$, and let $Y$ be a subset of $X$. We say that $Y$ is {\it preserved} in $F$ if there exists an element $B_i$ in $F$ with $i\in\{0,1, 2,\ldots,k\}$ such that $Y\subseteq B_i$. Throughout the article we will make heavy use of  the following theorem, referred to as the {\it chain preservation theorem (CPT)}.

\begin{theorem}
\label{thm:allchainsintact}
Let $T$ and $T'$ be two phylogenetic trees on $X$. 
Let $K$ be an (arbitrary) set of mutually taxa-disjoint chains that are common to $T$ and $T'$.
Then there exists a maximum agreement forest $F$ of $T$ and $T'$ such that 
\begin{enumerate}
\item every $n$-chain in $K$ with $n\geq 3$
is preserved in $F$, and
\item every 2-chain in $K$ 
that is pendant in at least one of $T$ and $T'$
is preserved in $F$.
\end{enumerate}
\end{theorem}

\noindent Following on from the last theorem, we say that common $n$-chains with $n \geq 3$, and common 2-chains that are pendant in at least one of $T$ and $T'$ are \emph{CPT-eligible} chains.  In our proofs, CPT-eligible chains
will function as `obstructions' that allow us to reason about the structure of maximum agreement forests.\\

\begin{figure}[t]
\center
\scalebox{1}{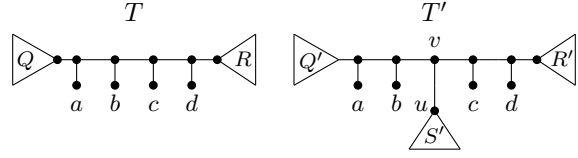}
\caption{Here $(a,b,c,d)$ is an interrupted 4-chain of $T$ and $T'$. Triangles indicates subtrees of $T$ and $T'$. The sets $S,Q,R,Q',R'$ are referred to in the proof of Theorem~\ref{t:lovely-chains}, which
proves that at least one maximum agreement forest of $T, T'$ does \emph{not} use the edge $\{u,v\}$ in $T'$. Note that $S'$ must contain at least one leaf, but the leaf set of any of the subtrees $Q$, $R$, $Q'$, and $R'$ may be empty, in which case $\{a,b\}$ or $\{c,d\}$ can become cherries.}
\label{fig:setup}
\end{figure}

We now turn to the second technical result whose proof is given in the appendix. Let $C=(a,b,c,d)$ be a 4-chain. We say that $C$ is an {\it interrupted 4-chain} of two phylogenetic trees $T$ and $T'$ on $X$ if $C$ is a chain of $T$ and, in $T'$, there exists a walk $p_a,p_b,v,p_c,p_d$ such that \purple{$p_b$}, $v$, and $p_c$ are three pairwise distinct vertices. \purple{Furthermore, the edge $e=\{u,v\}$ in $T'$ with \purple{$u\notin\{p_b,p_c\}$} is called the {\it interrupter} of $C$. Note that $v$ is not necessarily the parent of a leaf in $T'$, and that $T$ and the tree resulting from deleting $e$ in $T'$ and suppressing $v$ have $C$ as a common 4-chain.}
%Furthermore, by deleting the edge $e=\{u,v\}$ in $T'$ with $u\notin\{p_a,p_c\}$ and suppressing $v$, the resulting tree and $T$ have $C$ as a common 4-chain. We call $e$ the {\it interrupter} of $C$. 
An example of an interrupted 4-chain is shown in Figure~\ref{fig:setup}.

\begin{theorem}\label{t:lovely-chains}
Let $T$ and $T'$ be two \sidmablue{phylogenetic}
\stevenred{trees on $X$}, and let $C=(a,b,c,d)$ be an interrupted  4-chain of $T$ and $T'$. Then there exists a maximum agreement forest $F$ for $T$ and $T'$ such that, \purple{for each $B\in F$, $T'[B]$} does not use the interrupter of $C$ in $T'$.
\end{theorem}

\section{Main result and a bird's-eye view of the main arguments}

In this section, we state the main result of this paper and give an overview of our approach to establish it. The following lemma summarizes the situation after Reductions 1--7 have been applied to exhaustion and  is the foundation of Theorem~\ref{t:old-kernel}.
%For brevity we will simply write ``the underlying generator'' as shorthand for: a generator of a network $N$ that displays $T$ and $T'$ such that $d_\TBR(T,T')=r(N)$. 
\begin{lemma}
\label{lemma:foundations}
Let $T$ and  $T'$ be two phylogenetic trees on $X$ that cannot be reduced under Reductions 1--7. Let $G$ be the generator underlying a phylogenetic network $N$ on $X$ such that $r(N)=d_\TBR(T,T')$. Then, in obtaining $N$ from $G$, the following statements hold.
\begin{itemize}
\item[(a)] At most four taxa can be attached to each side of $G$.
\item[(b)] At most three taxa can be attached to each 0-breakpoint side of $G$.
\item[(c)] At most four taxa can be attached to each 1-breakpoint side of $G$ and only
sides of the form $1|3$ and $2|2$ can achieve this upper bound.
\item[(d)] At most four taxa can be attached to each  2-breakpoint side of $G$ and only
sides of the form $2|1|1$ can achieve this upper bound.
\end{itemize}
\end{lemma}
\begin{proof}
In~\cite[Lemma 7]{kelk2020new}, the authors showed that each 0-breakpoint side of $G$ has at most three taxa and that each other side of $G$ has at most four taxa. Consider a 1-breakpoint side $S$ of $G$. Suppose that four leaves get attached to $S$ in obtaining $N$ from $G$. If $S$ is a $0|4$ side, then $T$ and $T'$ have a common 4-chain, contradicting that Reduction 2 has been applied to exhaustion. Hence $S$ is either a $1|3$ or $2|2$ side. Next, consider a 2-breakpoint side $S$ of $G$. Again, suppose that four leaves get attached to $S$ in obtaining $N$ from $G$.  If $S$ is a $0||4$, $1||3$, or $2||2$ side on which the breakpoints of  $T$ and $T'$ coincide, then $T$ and $T'$ have a common subtree with at least two leaves, contradicting that Reduction 1 has been applied to exhaustion. Otherwise, if the breakpoints of $T$ and $T'$ do not coincide and $S$ is a $0|2|2$, $0|1|3$, $0|3|1$, or $0|4|0$ side, then it is straightforward to check that $T$ and $T'$ would have been reduced by Reduction 1, 4, 3 or 2
respectively, again a contradiction.
 It now follows that $S$ is a $2|1|1$ side. 
\qed
\end{proof}

Let $G$ be a $k$-generator as described in Lemma~\ref{lemma:foundations}. By~\cite[Lemma 1]{tightkernel}, $G$ has  $3(k-1)$ sides, and there are $2k$ breakpoints to divide across
these sides. Intuitively, after attaching the elements in $X$ to sides of $G$ in order to obtain $N$, we delete $k$ edges in $N$ to obtain a subdivision of $T$ and $k$ edges to obtain a subdivision of $T'$. Instead of deleting edges in $N$, we think of this as placing breakpoints on the sides of $G$. Now, with a view towards obtaining an upper bound on the number of leaves that can be attached to any  $k$-generator $G$, the best we can do after applying Reductions 1--7 to exhaustion is to have $2k$ 1-breakpoint sides with four taxa each, and $(k-3)$
0-breakpoint sides with three taxa each; this is the origin of the $11k-9$ kernel in \cite{kelk2020new}. The main bottleneck in achieving a kernel that is smaller than $11k-9$ are sides with four taxa and, in particular, those that only have one breakpoint. This explains the heavy emphasis on $1|3$, $2|2$ and $2|1|1$ sides in the rest of the article.\\

The high-level idea to achieve a kernel for $d_\TBR$ that is smaller than $11k-9$ is as follows. In Section~\ref{sec:new-rules}, we present  new Reductions 8, 9, and 10. 
\begin{enumerate}
\item A $1|3$  side triggers Reduction 8 that, as long as a `secondary' common 3-chain is available, reduces the number of taxa by one.  The `3' part in a $1|3$ side can itself function as a secondary chain for another $1|3$ side, so this reduction rule eliminates all but at most one $1|3$ side.
\item A $2|2$ side that (informally) 
has relatively many taxa on the adjacent sides can be reduced by Reduction 9, which essentially first transforms such a  side into a $1|3$ side before executing Reduction 8 if another $1|3$ side (and therefore a `secondary' common 3-chain) is available.
% is executed because a secondary common 3-chain is now definitely available.
\item A $2|1|1$ side that (informally) has relatively many taxa on the adjacent sides triggers the parameter-reducing Reduction 10.
\end{enumerate}

\noindent After applying Reductions 1--10  to exhaustion, all  sides with four taxa (apart from possibly one single exception) do \emph{not} have many taxa on the adjacent sides. This has the consequence that, once these sparse adjacent sides are taken into account, 4-taxa sides contribute (on average) significantly fewer than four taxa per side. With some  careful counting this leads to an improved kernel of size $9k-8$.\\

We are now in a position to state the main result of this paper. The proof is deferred until Section \ref{sec:alltogether}.

\begin{theorem}\label{t:main}
Let $T$ and $T'$ be two phylogenetic trees on $X$ that cannot be reduced under any of Reductions 1--10. If $d_\TBR(T,T')\geq 2$, then $|X| \leq 9d_\TBR(T,T')-8$.
\end{theorem}

We finish this section by noting that the portfolio of reduction rules should always be executed in the order $1, 2,\ldots,10$.  Every time a reduction rule executes, the sequence should be restarted from Reduction 1.  The fact that in all cases the number of taxa (and sometimes the TBR distance) is reduced by at least one, and the fact that the reduction rules themselves can be executed in polynomial-time, ensures polynomial-time execution overall.

\section{Three new reduction rules}\label{sec:new-rules}

\subsection{Reduction 8: A  reduction rule to reduce a $1|3$ side if there is a spare
common 3-chain available.}

\begin{figure}[t]
\center
\scalebox{1}{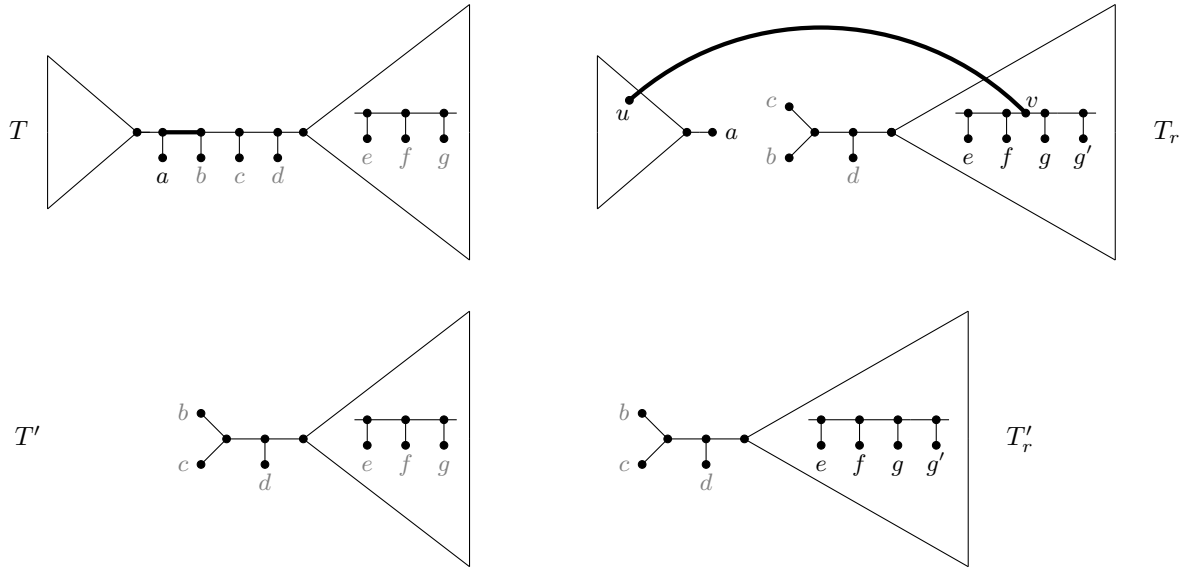}
\caption{Reduction 8A can be used to reduce a $1|3$ side $a|bcd$, as long as a secondary common 3-chain (here $\{e,f,g\}$) is available. Swapping the bold edge in $T$ with the bold edge in $T_r$ preserves $d_{\TBR}$ and creates a common pendant subtree $\{b,c,d\}$ which can be reduced under Reduction 1. 
%In total, this leads to a net reduction of one taxa. 
Observe that deleting the edge $\{p_a,p_b\}$ in $T$ disconnects $T$ into two smaller trees and, in the example above, 
%as shown above,
$(b,c,d)$ and $(e,f,g)$ are leaves of the same smaller tree. In general, this
%may no
does not need to be the case: $(b,c,d)$ and $(e,f,g)$ can be in different subtrees.}
% be the case.}}
\label{fig:redux8}
\end{figure}

The reduction rule that we describe first is designed to target the structures of two phylogenetic trees that are induced by a $1|3$ side $S= a|bcd$. We start by describing the first of~two parts of Reduction 8 and already note here that the second part is an application of Reduction~1.\\

\noindent {\bf Reduction 8A.} Let $T$ and $T'$ be two phylogenetic trees on $X$ that cannot be reduced under any of Reductions 1--7. Suppose that $T$ and $T'$ have two leaf-disjoint common 3-chains $C = (b,c,d)$ and $D=(e,f,g)$ such that $C$ is pendant in $T'$ with cherry $\{b,c\}$ and $C$ is not pendant in $T$, and there exists a taxon $a$ such that $(a,b,c,d)$ is a chain of $T$ and not a chain of $T'$. Then obtain $T_r$ from $T$ by extending $D$ to the  4-chain $(e,f,g,g')$ such that $g'\notin X$, deleting the edge $\{p_a,p_b\}$, suppressing $p_a$ and $p_b$, subdividing the edge $\{p_f,p_g\}$ with a new vertex $v$, and adding the edge $\{u,v\}$, where $u$ is a new vertex that subdivides an arbitrary edge 
of the  component that does not contain $e$ such that $(b,c,d)$ is a pendant 3-chain in $T_r$. Finally, obtain $T_r'$ from $T'$ by extending $D$ to the  4-chain $(e,f,g,g')$.  Then $T_r$ and $T_r'$ are two phylogenetic trees on $X\cup\{g'\}$. In what follows, we will call $D$ the {\it secondary common 3-chain} when executing Reduction 8A. An application of this reduction is shown in Figure~\ref{fig:redux8}.

\begin{lemma}
\label{lem:13kill}
Let $T$ and $T'$ be two phylogenetic trees on $X$ that cannot be reduced under any of Reductions 1--7.
If $T_r$ and $T_r'$ are two phylogenetic trees obtained from $T$ and $T'$, respectively, by  a single application of Reduction 8A, then $d_\TBR(T,T') = d_\TBR(T_r,T_r')$, and $T_r$ and  $T'_r$ have a common pendant subtree of size three.
\end{lemma}
\begin{proof}
We establish the theorem using the same notation as in the definition of Reduction 8A. Since $T$ and $T'$ cannot be reduced under any of Reductions 1--7, neither $C$ nor $D$ is the leaf set of a common subtree of $T$ and $T'$. Furthermore, by  Observation~\ref{obs:pendancy}, $D$ is pendant in at most one of $T$ and $T'$. Hence, without loss of generality, we assume that $\{f,g\}$ is not a cherry in either $T$ or $T'$. Now, let $S$ and $S'$ be the two phylogenetic trees obtained from $T$ and $T'$, respectively, by extending $D$ to the $4$-chain $(e, f, g, g')$. It follows from applying Lemma~\ref{l:7-reductions} to Reduction 2 that $d_\TBR(T,T')=d_\TBR(S,S')$. It remains to show that $d_\TBR(S,S') = d_\TBR(T_r,T_r')$. First, let $F=\{B_0,B_1,B_2,\ldots, B_k\}$ be a maximum agreement forest
for $S$ and $S'$. By CPT, we may assume that $C\subseteq B_i$ for some $i\in\{0, 1,2,\ldots,k\}$. Let $S_1$ and $S_2$ be the two phylogenetic trees obtained from $S$ by deleting the edge $e_1=\{p_a,p_b\}$ such that $S_1$ does not contain $b$. Since $C$ is pendant in $S'$, it follows that $B_i$ does not contain a taxon of $S_1$. This in turn implies that $e_1$ is not used by any embedding $S[B_j]$ with $B_j\in F\setminus \{B_i\}$. Hence $F$ is an agreement forest for $T_r$ and $T_r'$. Second, let $F_r=\{B_0,B_1,B_2,\ldots, B_k\}$ be a maximum agreement forest for $T_r$ and $T_r'$. By Theorem~\ref{t:lovely-chains}, we may assume that there exists no component $B_i\in F_r$ with $i\in\{0, 1,2,\ldots,k\}$ such that $T_r[B_i]$ uses the interrupter of $(e,f,g,g')$.
Hence $F_r$ is an agreement forest for $S$ and $S'$. Combining both cases, establishes that $d_\TBR(S,S') = d_\TBR(T_r,T_r')$. Moreover, by construction, $\{b,c,d\}$ is the leaf set of a common pendant subtree of $T_r$ and $T'_r$.
$\qed$
\end{proof}

We are now in a position to describe Reduction 8.\\

\noindent {\bf Reduction 8.} Let $T$ and $T'$ be two phylogenetic trees on $X$ that cannot be reduced under any of Reductions 1--7. If $T$ and $T'$ can be reduced under Reduction 8A, then reduce $T$ and $T'$ to $T_r$ and  $T'_r$, respectively, by an application of Reduction 8A followed by an application of Reduction 1.\\

\noindent If $T_r$ and $T_r'$ are obtained from $T$ and $T'$ as described in Reduction 8, we say that {\it Reduction 8 is applied to $C=(b,c,d)$ and $D=(e,f,g)$}, where $C$ and $D$ are as defined in Reduction 8A. \\

The next theorem shows that an application of Reduction 8 preserves the TBR distance and reduces the number of taxa by one. Furthermore, this reduction can be be executed in polynomial time  by trying all possible candidates for the taxa $\{a,b,c,d,e,f,g\}$ as defined in Reduction 8A.

\begin{theorem}
\label{t:shrink}
Let $T$ and $T'$ be two phylogenetic trees on $X$ that cannot be reduced under any of Reductions 1--7. Suppose that  $T$ and $T'$ can be reduced under Reduction 8A. Let $T_r$ and $T_r'$ be  two phylogenetic trees on $X'$ that are obtained from $T$ and $T'$, respectively, by a single application of Reduction 8.
Then $d_\TBR(T,T')=d_\TBR(T_r,T_r')$ and $|X'|=|X|-1$.
%Applying the subtree reduction to the trees $T_r, T'_r$ yields two trees that have the same TBR distance as the original trees $T, T'$ but on fewer taxa i.e. $|X|-1$ taxa.
\end{theorem}
\begin{proof}
Let $S$ and $S'$ be the two phylogenetic trees on $|X|+1$ leaves obtained from $T$ and $T'$, respectively, by a single application of Reduction 8A. By Lemma~\ref{lem:13kill}, we have $d_\TBR(T,T)=d_\TBR(S,S')$. Moreover, using the same notation as in the definition of Reduction 8A, it follows that $S$ and $S'$ have a common pendant subtree with leaf set $\{b,c,d\}$. Setting $T_r$ and $T_r'$ to be the two phylogenetic trees obtained from $S$ and $S'$, respectively, by applying Reduction 1 to $\{b,c,d\}$ and noting that $|X'|= |X|+1-2=|X|-1$ establishes the theorem.
\qed
\end{proof}

Reduction 8 also leads to the following observation, which we will need later.

\begin{observation}
\label{obs:13eateachother}
Let $N$ be a phylogenetic network on $X$ that displays two phylogenetic trees $T$ and $T'$ on $X$ that cannot be reduced under any of Reductions 1--8. Let $G$ be the generator that underlies $N$. Then $G$ has at most \underline{one} $1|3$ side.
\end{observation}
\begin{proof}
Suppose that there are two distinct such sides, $a|bcd$ and $.|efg$ where  ``.'' denotes a single taxon. Clearly, $\{b,c,d\} \cap \{e,f,g\} = \emptyset$ because the taxa are from distinct sides of $G$. Then $C=\{b,c,d\}$ and $D=\{e,f,g\}$ are two common 3-chains of $T$ and $T'$ that satisfy the three properties described in the definition of Reduction 8A. Hence,  $T$ and $T'$ can be further reduced under Reduction 8A and, therefore, under Reduction 8, a contradiction. \qed
\end{proof}

\subsection{Reduction 9: A reduction rule that triggers Reduction 8 by transforming certain $2|2$ sides into $1|3$ sides.}

The next  reduction rule targets the structures of two phylogenetic trees that are induced by a $2|2$ side $S = ab|cd$. We start by introducing an operation that does not reduce the number of leaves in the trees. Instead, this operation transforms certain $2|2$ sides of a generator into $1|3$ sides and is a \emph{p}recursor (hence the name P) to Reduction 9 that is described towards the end of this subsection.\\

\noindent {\bf Operation P.} Let $T$ and $T'$ be two phylogenetic trees on $X$ that cannot be reduced under any of Reductions 1--8. Suppose that $T$ has a non-pendant chain $(a,b,c,d)$, $T'$ has cherries $\{a,b\}$ and $\{c,d\}$, and there exists a maximum agreement forest $F$ for $T$ and $T'$ such that $\{a,b\}$ and $\{c,d\}$ are each preserved in $F$, but $\{a,b,c,d\}$ is not preserved in $F$. Then let $S=T$, and let $S'$ be the tree obtained from $T'$ by deleting $b$, suppressing $p_b$, subdividing the edge incident with $c$ with a new vertex $v$, and adding the edge $\{v,b\}$. \\

\noindent If $\{a,b,c,d\}$ satisfies all properties in the description of Operation P, we say that $\{a,b,c,d\}$ is {\it eligible for Operation P}. Moreover, if $S$ and $S'$ are obtained from $T$ and $T'$ as described above, we say that {\it Operation P is applied to $\{a,b,c,d\}$}. \\

\begin{theorem}
\label{thm:22to13}
Let $T$ and $T'$ be two phylogenetic trees on $X$ that cannot be reduced under any of Reductions 1--8. Furthermore, let  $S$ and $S'$ be two phylogenetic trees obtained from $T$ and $T'$, respectively, by applying Operation P to $\{a,b,c,d\}\subseteq X$. Then $d_{\TBR}(T, T') = d_{\TBR}(S,S')$. Moreover, $(b,c,d)$ is a common chain of $S$ and $S'$, and  $S'$ has cherry $\{b,c\}$.
\end{theorem}
\begin{proof}
We establish the theorem using the same notation as in the definition of Operation P. First, let $F=\{B_0,B_1,B_2,\ldots,B_k\}$ be a maximum agreement forest for $T$ and $T'$ such that $\{a,b\}$ and $\{c,d\}$ are each preserved in $F$, but $\{a,b,c,d\}$ is not preserved in $F$. Since the elements in $\{T[B_i]: i\in\{0,1,2,\ldots,k\}\}$ are pairwise vertex disjoint, no element $T[B_i]$ uses the edge $\{p_b,p_c\}$. Let $B_j$ be the element in $F$ such that $\{a,b\}\subseteq B_j$ and, similarly, let $B_{j'}$ be the element in $F$ such that $\{c,d\}\subseteq B_{j'}$. Then $$(F\setminus \{B_j,B_{j'}\})\cup \{B_j\setminus \{b\},B_{j'}\cup\{b\}\}$$ is an agreement forest for $S$ and $S'$ that has the same size as $F$. Hence $d_{\TBR}(S,S') \leq d_{\TBR}(T, T')$. Second, let $F=\{B_0,B_1,B_2,\ldots,B_k\}$ be a maximum agreement forest for $S$ and $S'$. By CPT, we may assume that $\{b,c,d\}\subseteq B_j$ for some $j\in\{0,1,2,\ldots,k\}$. Since $S|B_j=S'|B_j$ and the elements in $\{S[B_i]: i\in\{0,1,2,\ldots,k\}\}$ are pairwise vertex disjoint, it follows that the edge $\{p_a,p_b\}$ is not used by $S[B_i]$ for any $i\in\{0,1,\ldots,k\}$. Now let $B_{j'}$ be the element in $F\setminus \{B_j\}$ such that $a\in B_{j'}$. Then $$(F\setminus \{B_j,B_{j'}\})\cup \{B_j\setminus \{b\},B_{j'}\cup\{b\}\}$$ is an agreement forest for $T$ and $T'$ that has the same size as $F$. Thus $d_{\TBR}(S,S') \geq d_{\TBR}(T, T')$. Combining both cases establishes that $d_{\TBR}(S, S') = d_{\TBR}(T,T')$. Moreover, by construction of $S$ and $S'$ it follows immediately that $(b,c,d)$ is a common chain of $S$ and $S'$, and  $S'$ has cherry $\{b,c\}$.
\qed
\end{proof}

Let $T$ and $T$ be two phylogenetic trees on $X$. Following on from the description of Operation P, we present an explicit, polynomial-time algorithm---called Algorithm 1---in the appendix for testing whether or not, given a subset $\{a,b,c,d\}$ of $X$, there exists a maximum agreement forest $F$ for $T$ and $T'$ such that $\{a,b\}$ and $\{c,d\}$ are each preserved in $F$, but $\{a,b,c,d\}$ is not preserved in $F$. Although the algorithm does not necessarily catch \emph{all} situations when $F$ exists,  it is enough for our purposes. The high-level idea is that, as soon as a side $ab|cd$ has `many taxa on its surrounding sides', then  $T$ and $T'$ will contain easily detectable structures that constitute a certificate for the existence of $F$ and Algorithm 1 will find them.\\

The following corollary to Theorem \ref{thm:22to13} is useful later.

\begin{corollary}
\label{cor:preserved}
Let $T$ and $T'$ be two phylogenetic trees on $X$ that cannot be reduced under any of Reductions 1--8, and let  let $S$ and $S'$ be the two phylogenetic trees on $X$ that are obtained from $T$ and $T'$, respectively, by applying Operation P to $\{a,b,c,d\}\subseteq X$. Furthermore, let $N$ be a phylogenetic network on $X$ that displays $T$ and $T'$. If the generator $G$ that underlies $N$ has a $2|2$ side $S=ab|cd$, then $N$ also displays $S$ and $S'$.
\end{corollary}

\begin{proof}
Using the same notation as in the definition of Operation P, observe that the breakpoint on $S$ is relative to $T'$. To see that $N$ also displays $S$ and $S'$, we view $S$ as the $1|3$ side $a|bcd$, where the breakpoint is now relative to $S'$. 
\qed
\end{proof}

Theorem \ref{thm:22to13} and Corollary \ref{cor:preserved} are the theoretical foundation for Operation P. Once Operation P is applied, Reduction 8 may be triggered if another $1|3$ side is available. This can happen in two slightly different ways which we describe next as Reduction 9.1 and 9.2. Reduction 9.1 is tried first and, if it fails, Reduction 9.2 is tried. Essentially Reduction 9.1 converts one $2|2$ side into a $1|3$ side and Reduction 9.2 converts two $2|2$ sides into $1|3$ sides. In both cases, the new $1|3$ sides trigger Reduction 8. Let $T$ and $T'$ be two phylogenetic trees on $X$ that cannot be reduced under any of Reductions 1--8. Then Reductions 9.1 and 9.2 are defined as follows. \\

\noindent  \textbf{Reduction 9.1.} Suppose that there exists $\{a,b,c,d\} \subseteq X$ such that $C=(b,c,d)$ is a common chain of $T$ and $T'$, $C$ is pendant in $T'$ with cherry $\{b,c\}$ and not pendant in $T$, and $(a,b,c,d)$ is a chain of $T$ and not a chain of $T'$. Suppose furthermore that
there exists $\{a',b',c',d'\} \subseteq X$ which is eligible for Operation P, where $\{a,b,c,d\} \cap \{a',b',c',d'\} = \emptyset$. Then an application of Reduction 9.1 to $T$ and $T'$ consists of an application of  Operation P to $\{a',b',c',d'\}$, thereby creating two phylogenetic trees with a common 3-chain $(b',c',d')$, and a subsequent application of Reduction 8  to $C$ and the newly created secondary common 3-chain $D=(b',c',d')$.\\

\noindent \textbf{Reduction 9.2.}  Suppose that there exist two disjoint subsets  
$\{a',b',c',d'\}$ and $\{a'',b'',c'',d''\}$ of $X$, such that both are eligible for Operation P. Then, an application of Reduction 9.2 to $T$ and $T'$ consists of an  application of Operation P to $\{a',b',c', d'\}$ followed by an application of the same operation to $\{a'', b'', c'', d''\}$ if it is still eligible\footnote{In our analysis later in the article, we apply Reduction 9.2 in a situation where $\{a'', b'', c'', d''\}$ is definitely still eligible for transformation after $\{a',b',c', d'\}$ has been transformed.}, and finally an application of Reduction 8 to $C=\{b',c',d'\}$ and the secondary common chain $D=(b'',c'',d'')$.\\

\noindent It is important to note that Reduction 9.2 is an `all-or-nothing' reduction, i.e. it either executes fully or not at all. Specifically, it does not execute the first application of  Operation P but not the second. As Algorithm 1 (see appendix) runs in polynomial time, it follows that Reductions 9.1 and 9.2 can be executed in polynomial time  by trying all possible candidates for the taxa $\{a,b,c,d, a',b',c',d'\}$ and $\{a',b',c',d',a'', b'', c'', d''\}$, respectively. Lastly, since we do not always need to distinguish between Reduction 9.1 and Reduction 9.2, we refer to an application of one of the two reductions as {\it Reduction 9}.

\subsection{Reduction 10: A reduction rule to reduce
%(in a parameter-reducing way) 
certain $2|1|1$ sides.}

The last new reduction rule targets the structures of two phylogenetic trees that are induced by a $2|1|1$ side $S = ab|c|d$. Reduction 10 is much more straightforward than Reductions 8 and 9. \\

\noindent {\bf Reduction 10.} Let $T$ and $T'$ be two phylogenetic trees on $X$ that cannot be reduced under any of Reductions 1--9. If $T$ has two cherries $\{a,b\}$ and $\{c,d\}$, $T'$ has the 3-chain $(a,b,c)$ such that $\{b,c\}$ is a cherry, and there exists a maximum agreement forest $F$ for $T$ and $T'$  such that $\{c\}\in F$, then reduce $T$ and $T'$ to $T_r=T|X\setminus \{c\}$ and $T_r'=T'|X\setminus\{c\}$, respectively.\\

\noindent  If $\{a,b,c,d\}$ satisfies all properties in the description of Reduction 10, we say that $\{a,b,c,d\}$ is {\it eligible for Reduction 10}. Moreover, if $T_r$ and $T_r'$ are obtained from $T$ and $T'$ as described above, we say that {\it Reduction 10 is applied to $\{a,b,c,d\}$}.\\

The next theorem shows that Reduction 10 is parameter reducing. Its proof is straightforward and omitted.

\begin{theorem}
\label{thm:211eligible}
Let $T$ and $T'$ be two phylogenetic trees on $X$ that cannot be reduced under any of Reductions 1--9. Furthermore, let  $T_r$ and $T'_r$ be two phylogenetic trees obtained from $T$ and $T'$, respectively, by a single application of Reduction 10. Then $d_{\TBR}(T_r, T'_r) = d_{\TBR}(T,T') - 1$.
\end{theorem}

\noindent It remains to establish that Reduction 10 can be executed in polynomial time. In the appendix, we present an explicit, polynomial-time algorithm---called Algorithm 2---for testing whether there exists a maximum agreement forest $F$ for $T$ and $T'$  such that $\{c\}\in F$. 
As Algorithm 2  runs in polynomial time, it follows that Reduction 10 can be executed in polynomial time  by trying all possible candidates for the taxa $\{a,b,c,d\}$ as defined in Reduction 10.
Moreover, an application of Reduction 10 decreases the number of taxa and the TBR distance  both by exactly 1.

\section{A win-win scenario}\label{sec:MAFs-everywhere}

In this section we explore the interplay of sides of a generator that are adjacent to each other. We will see that a generator side whose adjacent sides are densely decorated with taxa \stevenred{triggers} reduction rules and that a generator side that does not trigger a reduction rule has adjacent sides that are, on average, only sparsely decorated with taxa. To this end, we establish several results that pinpoint when a subset of taxa is eligible for Operation P or Reduction 10. We begin with a key insight.

\begin{observation}
\label{obs:chainseverywhere}
Let $T$ and  $T'$ be two phylogenetic trees on $X$ that cannot be reduced under Reductions 1--7. Let $G$ be the generator underlying a phylogenetic network $N$ on $X$ such that $r(N)=d_\TBR(T,T')$, and let $S$ be a side of $G$. If at least three taxa are attached to $S$ in obtaining $N$ from $G$, then $T$ and $T'$ have a CPT-eligible chain unless $S$ is a $1|1|1$ side.
\end{observation}
\begin{proof}
If $S$ is a 0-breakpoint side, then it immediately follows that $T$ and $T'$ have a common 3-chain that is CPT eligible. Suppose that $S$ is a 1-breakpoint side. Then $S$ is a $n_1|n_2$ side, where $n_1\geq 0$ and $n_2\geq 0$ denote the number of taxa attached to $S$ on either side of the breakpoint. Since $n_1+n_2\geq 3$, either $n_1\geq 2$ or $n_2\geq 2$. Hence $T$ and $T'$ have a common 2-chain that is pendant in one of $T$ and $T'$ and therefore CPT eligible. Lastly, suppose that $S$ is a 2-breakpoint side. Similar to the 1-breakpoint case, $S$ is a  $n_1|n_2|n_3$ side, where $n_1\geq 0$, $n_2\geq 0$, and $n_3\geq 0$ denote the number of taxa attached to $S$ before the first breakpoint, after the first and before the second breakpoint, and after the second breakpoint, respectively. Since $n_1+n_2+n_3\geq 3$ and $S$ is not a $1|1|1$ side it again follows that $T$ and $T'$ have a common 2-chain that is pendant in one of $T$ and $T'$ and therefore CPT eligible. \qed
\end{proof}

In the remainder of this section, we carefully analyze $2|2$ and $2|1|1$ sides, and establish sufficient conditions under which a subset of taxa that decorates such a side is eligible for Operation P or Reduction 10. Let $S$ be a side of a generator $G$. Viewing $G$ as a graph, $S$ can either be a \emph{simple edge}, i.e. an edge  that is not part of a multi-edge, an edge that is part of a \emph{multi-edge}, or a \emph{loop}. A multi-edge of $G$ contains at most two edges, due to the fact that each vertex of $G$ has degree three. The only exception is if $G$ has exactly two vertices, and one multi-edge consisting of three edges. This  implies that $d_{\TBR} \leq 2$. By assuming throughout the rest of the paper that $d_{\TBR} \geq 3$ we can exclude this case\footnote{This does not harm the final upper bound $9k-8$ on the size of the kernel, because we can easily test in polynomial time whether $d_{\TBR}(T,T') \leq 2$ and if so exactly compute $d_{\TBR}$ in the same time bound. Subsequently we can output a trivial YES/NO instance to complete the kernelization.}. The following analyses depend on whether $S$ is a simple edge, an edge of a multi-edge, or a loop.

Observe that a $2|2$ \stevenred{(or a $1|3$)} side cannot be a loop because the phylogenetic  tree that does not have a breakpoint on that side would contain a cycle.

\subsection{$2|2$ sides}

\begin{figure}[t]
\center
\scalebox{1}{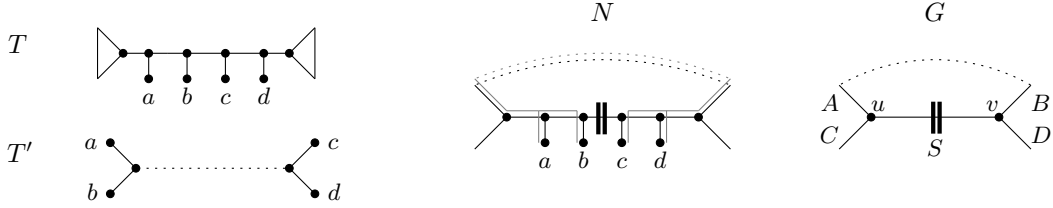}
\caption{The situation described in Theorem \ref{thm:22sparse}, which concerns $2|2$ sides $S=ab|cd$, where $S$ is not part of a multi-edge. The \revision{gray} edges in $N$ indicate an image of $T'$.
}
\label{fig:22simple}
\end{figure}

\begin{theorem}
\label{thm:22sparse}
Let $T$ and  $T'$ be two phylogenetic trees on $X$ that cannot be reduced under Reductions 1--7. Let $G$ be the generator underlying a phylogenetic network $N$ on $X$ that displays $T$ and $T'$ such that $r(N)=d_\TBR(T,T')$. Furthermore, let $S =ab|cd$ be a side of $G$ that is a simple edge $\{u,v\}$. Let $A$, $B$, $C$, and $D$ be the four sides incident with $S$ such that $A$ and $C$ are both incident with $u$, and $B$ and $D$ are both incident with $v$. If each of $A$ and $C$ is decorated with at least two taxa, or each of $B$ and $D$ is decorated with at least two taxa in obtaining $N$ from $G$, then $\{a,b,c,d\}$ is eligible for Operation P.
\end{theorem}

\begin{proof}
%We prove this by showing that, if the sides $A, B$ do not exist, then Reduction 9 can definitely be triggered, contradicting the assumption that reduction rules have been applied exhaustively.
Assume without loss of generality that $T'$ has cherries $\{a,b\}$ and $\{c,d\}$ and that $T$ has a non-pendant chain $(a,b,c,d)$ as illustrated in Figure \ref{fig:22simple}. Then $(a,b)$ and $(c,d)$ are two common 2-chains of $T$ and $T'$. Both of these chains are pendant in $T'$ and therefore CPT eligible. %Furthermore, since $S$ is a simple edge, there exist two distinct sides $A$ und $C$ that are incident to $u$ and two distinct sides $B$ and $D$ that are incident to $v$. 
For each $Y\in\{A,B,C,D\}$, let $P_Y$ be the path associated with $Y$ in $N$.
%\footnote{Some of $A$, $B$, $C$, and $D$ may be part of a  multi-edge. However, none of them is a loop.} 
Now, consider an image $I$ of $T'$ in $N$. 
Let $P$ be the path from $a$ to $c$ in $I$.
%Let $P$ be the path of $I$ that corresponds to the path from $p_a$ to $p_c$ in $T'$. 
Since $S$ is a 1-breakpoint side, one of $P_A$ and $P_C$ is a subpath of $P$ and, similarly, one of  $P_B$ and $P_D$ is a subpath of $P$. We assume without loss of generality that $P_A$ and $P_B$ are both subpaths of $P$. It follows that $T'$ has no breakpoint on $A$ or $B$ and, hence each of $A$ and $B$ has at most one breakpoint relative to $T$. Moreover, by the assumption in the statement of the theorem, at least one of $A$ and $B$ is decorated with at least two taxa in the process of obtaining $N$ from $G$. 

Suppose that $A$ is is decorated with at least three taxa. Since $A$ has at most one breakpoint,  it follows from Observation \ref{obs:chainseverywhere} that $T$ and $T'$ have a CPT-eligible chain $Z$ whose elements are attached to $A$ in obtaining $N$ from $G$. Let $K=\{\{a,b\},\{c,d\},Z\}$.  By applying the CPT to $K$, there exists a maximum agreement forest $F$ for $T$ and $T'$ such that each element in $K$ is preserved in $F$. Assume that there exists an element $B$ in $F$ such that $\{a,b,c,d\}\subseteq B$. Let $B'$ be the element in $F$ such that $Z\subseteq B'$. If $B\ne B'$, then $T'[B]$ and $T'[B']$ are not vertex disjoint, a contradiction. Hence $B=B'$. But then $T|(\{a,b,c,d\}\cup Z\})\ne T'|(\{a,b,c,d\}\cup Z\})$, another contradiction. It follows that $B$ does not exist and $\{a,b,c,d\}$ is eligible for Operation P.  An identical analysis holds for when $B$ is decorated with at least three taxa. 

We can now assume that neither $A$ nor $B$ is decorated with at least three taxa. Then, by the statement of the theorem,  $A$ or $B$ is decorated with exactly two taxa. We establish the theorem for when $A$ is decorated with exactly two taxa. An analogous and symmetric argument holds for when $B$ is decorated with exactly two taxa. Let $e$ and $f$ be the two taxa that are attached to $A$ in obtaining $N$ from $G$ such that the path from $p_e$ to $p_a$ in $T'$ does not pass through $p_f$. Intuitively, $e$ is closer than $f$ to $a$ in $N$. If $A=|ef$ or $A=ef|$, then $T$ and $T'$ have a common 2-chain $(e,f)$ that is pendant in $T$. By applying an argument that is similar to that of the last paragraph and setting $Z=\{e,f\}$, we deduce that $\{a,b,c,d\}$ is eligible for Operation P.  Hence $A=ef$ or $A=e|f$. Let $F$ be a maximum agreement forest for $T$ and $T'$. Assume that there exists an element $B$ in $F$ such that $\{a,b,c,d\}\subseteq B$. Since $T|B=T'|B$, we have $B=\{a,b,c,d\}$. Furthermore, each of $e$ and $f$ is a singleton in $F$. We now consider two cases, depending on whether $A$ has one or zero breakpoints, and show that there exists another maximum agreement forest for $T$ and $T'$ that has the desired properties such that $\{a,b,c,d\}$ is eligible for Operation P.

First, suppose that $A=ef$.  Let $$F'=(F\setminus \{B,\{e\},\{f\}\})\cup \{\{a,b,e,f\},\{c,d\}\}$$ be a forest. Noting that $|F'|< |F|$, it follows by the maximality of $F$ that $F'$ is not an agreement forest for $T$ and $T'$. Hence, by construction of $F'$, there exists an element $B'$ in $F'\setminus \{\{a,b,e,f\}\}$ such that $T[B']$ uses the edge $\{p_e,p_f\}$ in $T$.  Let $B_1',B_2'$ be a bipartition of $B'$ such that neither $T[B_1']$ nor $T[B_2']$ uses the edge $\{p_e,p_f\}$ in $T$. As $B'$ is also an element of $F$, it now follows that  $$F''=(F\setminus \{B,B',\{e\},\{f\}\})\cup \{\{a,b,e,f\},\{c,d\},B_1',B_2'\}$$ is an agreement forest for $T$ and $T'$ with $|F|=|F''|$ and in which $\{a,b\}$ and $\{c,d\}$ are both preserved, and $\{a,b,c,d\}$ is not preserved. Hence, $\{a,b,c,d\}$ is eligible for  Operation P. 

Second, suppose that $A=e|f$. 
%As the breakpoint on $A$ is relative to $T$, 
Observe that $(e,a,b,c,d)$ is a chain of $T$. Let $$F'=(F\setminus \{B,\{e\}\})\cup \{\{a,b,e\},\{c,d\}\}$$ be a forest. Since there exists no element in $F\setminus \{B\}$ whose embedding in $T$ uses $p_e$, $F'$ is an agreement forest for $T$ and $T'$. As $|F'|=|F|$ it now follows again that $\{a,b,c,d\}$ is eligible for Operation P. 
\qed
\end{proof}

\begin{theorem}
\label{thm:22multiedge}
Let $T$ and  $T'$ be two phylogenetic trees on $X$ that cannot be reduced under Reductions 1--7. Let $G$ be the generator underlying a phylogenetic network $N$ on $X$ that displays $T$ and $T'$ such that $r(N)=d_\TBR(T,T')$. Furthermore, let $S =ab|cd$ be a side of $G$ that is part of a multi-edge $\{u,v\}$. Let $A$, $B$, and $M$ be the three sides incident with $S$ such that $M$ is incident with $u$ and $v$, $A$ is only incident with $u$, and $B$ is only incident with $v$. Then $\{a,b,c,d\}$ is eligible for Operation P if each of the following conditions hold in obtaining $N$ from $G$:
\begin{enumerate}
\item one of $A$ and $B$ is decorated with at least two taxa; and 
\item $M$ is decorated with at least one taxon.
\end{enumerate} 
\end{theorem}
\begin{proof}
Assume without loss of generality that $T'$ has cherries $\{a,b\}$ and $\{c,d\}$ and that $T$ has a non-pendant chain $(a,b,c,d)$ as illustrated in Figure \ref{fig:22multiedge}. Then $(a,b)$ and $(c,d)$ are two common 2-chains of $T$ and $T'$. Both of these chains are pendant in $T'$ and therefore CPT eligible. For each $Y\in\{A,B,M\}$, let $P_Y$ be the path associated with $Y$ in $N$. Now, consider an image $I$ of $T'$ in $N$. 
Let $P$ be the path from $a$ to $c$ in $I$.
%Let $P$ be the path of $I$ that corresponds to the path from $p_a$ to $p_c$ in $T'$. 
Then either $P_A$ and $P_B$ are subpaths of $P$, or $P_M$ is a subpath of $P$. If $P_A$ and $P_B$ are subpaths of $P$ then, since the first condition in the statement of the theorem is satisfied, we can apply the same argument as in the proof of Theorem \ref{thm:22sparse} to establish that $\{a,b,c,d\}$ is eligible for Operation P. We may therefore assume that $P_M$ is a subpath of $P$. Since $T$ does not contain a cycle and the breakpoint on $S$ is relative to $T'$, it follows that $M$ has a single breakpoint that is relative to $T$.
\begin{figure}[t]
\center
\scalebox{1}{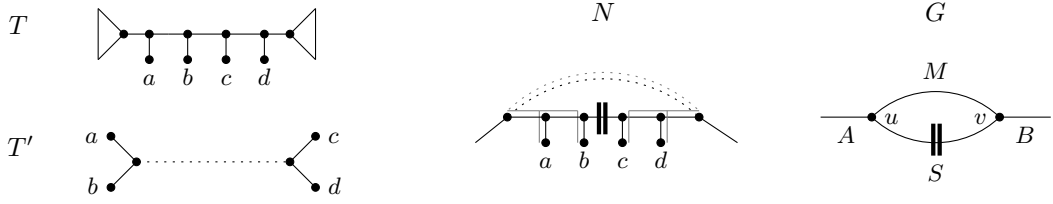}
\caption{The situation described in Theorem~\ref{thm:22multiedge}, which concerns $2|2$ sides $S=ab|cd$, where $S$ is  part of a multi-edge. The \revision{gray} edges in $N$ indicate an image of $T'$ that passes through the path in $N$ that is associated with $M$.
%The situation described in Theorem \ref{thm:22multiedge}, which concerns $2|2$ sides $ab|cd$ that are part of a multi-edge. The grey tree is the image of the path from $\{a,b\}$ to $\{c,d\}$ in $T'$ in the case that this image passes through $M$, which is the other side in the multi-edge. The high-level idea is that if $M$ has more than 0 taxa, then $\{a,b,c,d\}$ is definitely eligible for transformation via Theorem \ref{thm:22to13} (and can then participate in Reduction 9).
}
\label{fig:22multiedge}
\end{figure}

Let $Z$ be the set of taxa that is attached to $M$ in obtaining $N$ from $G$. First, assume that $Z$ is CPT eligible. Then $|Z|\geq 2$, and there exists a maximum agreement forest $F$ for $T$ and $T'$ that preserves each element in $\{\{a,b\},\{c,d\},Z\}$. If there exists an element $B$ in $F$ such that $\{a,b,c,d\}\subseteq B$, then $B$ also contains $Z$ since, otherwise, $T'[B]$ and $T'[B']$ are not vertex disjoint, where $B'$ is the element in $F\setminus \{B\}$ such that $Z\subseteq B'$. Hence $B=B'$, thereby implying that $T[B]\ne T'[B]$, a contradiction. It now follows that $\{a,b,c,d\}$ is not a subset of any element in $F$ and, thus, $\{a,b,c,d\}$ is eligible for Operation P. Second, assume that $Z$ is not CPT eligible. Since the second condition in the statement of the theorem is satisfied, it follows from the fact that $M$ is a 1-breakpoint side and from the contrapositive of Observation \ref{obs:chainseverywhere} that $1\leq |Z|\leq 2$. Let $F$ be a maximum agreement forest that preserves $\{a,b\}$ and $\{c,d\}$. Again assume that there exists an element $B$ in $F$ such that $\{a,b,c,d\}\subseteq B$.  We next consider two cases. 

First suppose that $Z=\{e,f\}$.  Since $Z$ is not CPT eligible, it follows that $M=e|f$. Without loss of generality, we assume that the path from $p_a$ to $p_e$ in $T'$ does not pass through $p_f$. Since $T|B=T'|B$, it follows that $\{e\}$ and $\{f\}$ are  elements in $F$. Observe that $(e,a,b,c,d,f)$ is a chain of $T$.
Now, let $$F'=(F\setminus \{\{B,\{e\},\{f\}\})\cup \{\{a,b,e\},\{c,d,f\}\}.$$ As $F$ is an agreement forest for $T$ and $T'$ and each edge of $P$ is used by $T'[B]$, $F'$ is such a forest as well, contradicting the minimality of $F$. Hence $B$ does not exist in $F$ and $\{a,b,c,d\}$ is therefore eligible for Operation P.

Second suppose that $Z=\{e\}$. Then $M=|e$ or $M=e|$. Since $T|B=T'|B$, it follows that $\{e\}$ is an element in $F$. Observe that either $(e,a,b,c,d)$ or $(a,b,c,d,e)$ is a chain of $T$.
%and that $p_a,w,p_e,w',p_c$ is a path of $T'$ whose edges are used by $T'[B]$.  
Now, if $(e,a,b,c,d)$  is a chain in $T$, let $$F'=(F\setminus \{B,\{e\}\})\cup \{\{a,b,e\},\{c,d\}\}$$ and, if $(a,b,c,d,e)$  is a chain in $T$, let $$F'=(F\setminus \{B,\{e\}\})\cup \{\{a,b\},\{c,d,e\}\}.$$ As $F$ is a maximum agreement forest for $T$ and $T'$, it follows that, regardless of which case applies, $F'$ is also such a forest. Thus, $\{a,b,c,d\}$ is eligible for Operation P. \qed 
\end{proof}

\subsection{$2|1|1$ sides}   

\begin{theorem}
\label{thm:211sparse}
Let $T$ and  $T'$ be two phylogenetic trees on $X$ that cannot be reduced under Reductions 1--7. Let $G$ be the generator underlying a phylogenetic network $N$ on $X$ that displays $T$ and $T'$ such that $r(N)=d_\TBR(T,T')$. Furthermore, let $S =ab|c|d$ be a side of $G$ that is a simple edge $\{u,v\}$. Let $A$, $B$, $C$, and $D$ be the four sides incident with $S$ such that $A$ and $C$ are both incident with $u$, and $B$ and $D$ are both incident with $v$. If each of $A$ and $C$ is decorated with at least two taxa, or each of $B$ and $D$ is decorated with at least two taxa in obtaining $N$ from $G$, then $\{a,b,c,d\}$ is eligible for Reduction 10.
%Assume that Reductions 1-7 have been applied exhaustively. Let $S =ab|c|d$ be a $2|1|1$ side of the underlying generator that is not part of a multi-edge. Suppose the following does \underline{\emph{not}} hold: there exist two distinct sides $A$ and $B$ that are both incident to $S$, but to different ends of $S$, whereby $A$ has at most 1 taxon and $B$ has at most 1 taxon.  Then $\{a,b,c,d\}$ is eligible for reduction via Theorem \ref{thm:211eligible}.
\end{theorem}

\begin{proof}
Assume without loss of generality that $T$ has cherries $\{a,b\}$ and $\{c,d\}$ and that $T'$ has a pendant 3-chain $(a,b,c)$ with cherry $\{b,c\}$ as illustrated in Figure \ref{fig:211simple}. Then the 2-chain $(a,b)$ is CPT eligible. Let $F$ be a maximum agreement forest for $T$ and $T'$ such that $\{a,b\}$ is preserved in $F$. Let $B$ be the element in $F$ with $\{a,b\}\subseteq F$. Then either $\{c\}\in F$ or $\{a,b,c\}\subseteq B$. Assume that the latter holds. Then, as $T|B=T|B'$, we have $B=\{a,b,c\}$. We freely use this observation throughout the rest of the proof.

Now, for each $Y\in\{A,B,C,D\}$, let $P_Y$ be the path associated with $Y$ in $N$.
Furthermore, let $I$ be an image of $T$ in $N$, and let $P$ be the path from $a$ to $c$ in $I$.  Since $S$ is a 2-breakpoint side, either $P_A$ or $P_C$ is a subpath of $P$ and, similarly, one of  $P_B$ or $P_D$ is a subpath of $P$. We assume without loss of generality that $P_A$ and $P_B$ are both subpaths of $P$. If follows that $T$ has no breakpoint on $A$ or $B$ and, hence each of $A$ and $B$ has at most one breakpoint relative to $T$.  

Suppose that $A$ is  a 1-breakpoint side $A=|ef$ or $A=ef|$ that is decorated with exactly two taxa $e$ and $f$, or that $A$ is decorated with at least three taxa. Since $A$ has at most one breakpoint, it follows from Observation~\ref{obs:chainseverywhere} that $T$ and $T'$ have a CPT-eligible chain $Z$ whose elements are attached to $A$ in obtaining $N$ from $G$. Let $K=\{\{a,b\},Z\}$.  By applying the CPT to $K$, there exists a maximum agreement forest $F'$ for $T$ and $T'$ such that each element in $K$ is preserved in $F'$. Let $B$ and $B'$ be the elements of $F'$ such that $\{a,b\}\subseteq B$ and $Z\subseteq B'$. Assume that $\{c\}\notin F'$. Then, by the observation in the first paragraph of the proof, we have $B=\{a,b,c\}$. Since $T[B]$ and $T[B']$ are vertex disjoint, it follows that $B=B'$. In turn, this implies that $T|(\{a,b,c\}\cup Z) \ne T'|(\{a,b,c\}\cup Z)$, thereby contradicting that $F'$ is an agreement forest for $T$ and $T'$. Hence $\{c\}\in F'$ and, so, $\{a,b,c,d\}$ is eligible for Reduction 10. An identical analysis holds for when $B$ is decorated with at least three taxa. Hence, one of $A$ and $B$ is decorated with exactly two taxa. 

Now, reconsider $F$. If $\{c\}\in F$, then $\{a,b,c,d\}$ is clearly eligible for Reduction 10. We may therefore assume that $B=\{a,b,c\}$ and, consequently, $\{d\}\in F$. We next distinguish two cases and show that there always exists another maximum agreement forest for $T$ and $T'$ that has the desired property such that $\{a,b,c,d\}$ is eligible for Reduction 10. 
\begin{enumerate}
\item [(1)] Suppose that $A$ is decorated with exactly two taxa $e$ and $f$ such that the path from $p_e$ to $p_a$ in $T$ does not pass through $p_f$. By the definition of an agreement forest, it follows that $\{e\}$ and $\{f\}$ are elements of $F$.  Recall that $A$ has at most one breakpoint and that this breakpoint is, if it exists, relative to $T'$.
%If $A=|ef$ or $A=ef|$, then $T$ and $T'$ have a common 2-chain $(e,f)$ that is pendant in $T'$. By applying an argument that is similar to that of the last paragraph and setting $Z=\{e,f\}$, we deduce that $\{a,b,c,d\}$ is eligible for Reduction 10. 
Hence $A$ is either a 0-breakpoint side $A=ef$ or a 1-breakpoint side $A=e|f$. First, if $A=ef$, let $$F'=(F\setminus \{B,\{e\},\{f\}\})\cup \{\{a,b,e,f\},\{c\}\}$$ be a forest. Noting that $|F'|< |F|$, it follows by the maximality of $F$ that $F'$ is not an agreement forest for $T$ and $T'$. Hence, by construction of $F'$, there exists an element $B'$ in $F'\setminus \{\{a,b,e,f\}\}$ such that $T'[B']$ uses the edge $\{p_e,p_f\}$ in $T'$.  Let $B_1',B_2'$ be a bipartition of $B'$ such that neither $T'[B_1']$ nor $T'[B_2']$ uses the edge $\{p_e,p_f\}$ in $T'$. As $B'$ is also an element of $F\setminus \{B,\{e\},\{f\}\}$, it now follows that  $$F''=(F\setminus \{B,B',\{e\},\{f\}\})\cup \{\{a,b,e,f\},\{c\},B_1',B_2'\}$$ is another maximum agreement forest for $T$ and $T'$ in which $\{c\}$ is a singleton. Hence, $\{a,b,c,d\}$ is eligible for Reduction 10. Second, if $A=e|f$, let $$F'=(F\setminus \{B,\{e\}\})\cup \{\{a,b,e\},\{c\}\}$$ be a forest. Since there exists no element in $F\setminus \{B\}$ whose embedding in $T'$ uses $p_e$, $F'$ is another maximum agreement forest for $T$ and $T'$. It follows again that $\{a,b,c,d\}$ is eligible for Reduction 10. 
%\marginpar{Have not merged Cases (1) and (2) because of the special role of $d$.}

\item [(2)] Suppose that $B$ is decorated with exactly two taxa $e$ and $f$ such that the path from $p_e$ to $p_d$ in $T$ does not pass through $p_f$. As in Case (1), $\{e\}$ and $\{f\}$ are elements of $F$. 
%Moreover, if $B=|ef$ or $B=ef|$, then $T$ and $T'$ have a common 2-chain $(e,f)$ that is pendant in $T'$. By applying an argument that is similar to that of the third paragraph of the proof and setting $Z=\{e,f\}$, we deduce that $\{a,b,c,d\}$ is eligible for Reduction 10. 
Moreover, $B$ is either a 0-breakpoint side $B=ef$ or a 1-breakpoint side $B=e|f$, where the breakpoint is relative to $T'$. First, if $B=ef$, let $$F'=(F\setminus \{B,\{d\},\{e\},\{f\}\})\cup \{\{a,b\},\{c\},\{d,e,f\}\}$$ be a forest. Noting that $|F'|< |F|$, it follows by the maximality of $F$ that $F'$ is not an agreement forest for $T$ and $T'$. Hence, by construction of $F'$, there exists an element $B'$ in $F'\setminus \{\{d,e,f\}\}$ such that $T'[B']$ uses the edge $\{p_e,p_f\}$ in $T'$.  Let $B_1',B_2'$ be a bipartition of $B'$ such that neither $T'[B_1']$ nor $T'[B_2']$ uses the edge $\{p_e,p_f\}$ in $T'$. As $B'$ is also an element of $F\setminus \{B,\{d\},\{e\},\{f\}\}$, it now follows that  $$F''=(F\setminus \{B,B',\{d\},\{e\},\{f\}\})\cup \{\{a,b\},\{c\},\{d,e,f\},B_1',B_2'\}$$ is another maximum agreement forest for $T$ and $T'$ in which $\{c\}$ is a singleton. Hence, $\{a,b,c,d\}$ is eligible for Reduction 10. Second, if $B=e|f$, let $$F'=(F\setminus \{B,\{d\},\{e\}\})\cup \{\{a,b\},\{c\},\{d,e\}\}.$$  Since there exists no element in $F\setminus \{B\}$ whose embedding in $T'$ uses $p_e$, $F'$ is another maximum agreement forest for $T$ and $T'$. Thus $\{a,b,c,d\}$ is eligible for Reduction 10. \qed
\end{enumerate}
\end{proof}

\begin{figure}[t]
\center
\scalebox{1}{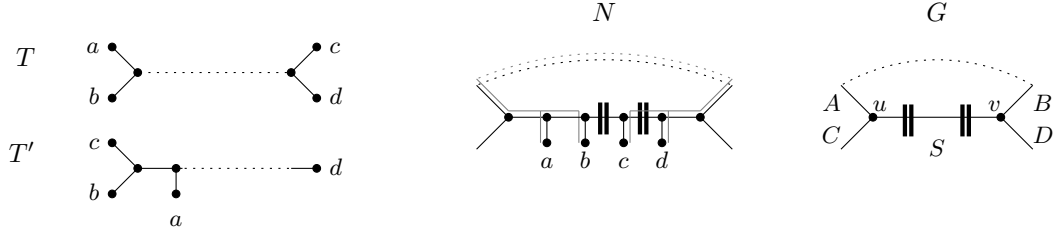}
\caption{The situation described in Theorem \ref{thm:211sparse}, which concerns $2|1|1$ sides $S=ab|c|d$, where $S$ is not part of a multi-edge. The \revision{gray} edges in $N$ indicate an image of $T$.
%The situation described in Theorem \ref{thm:211sparse}, which concerns $2|1|1$ sides $ab|c|d$. The grey tree is the image of the path from $\{a,b\}$ to $\{c,d\}$ in $T$. The high-level idea is that if at least one of sides $A, B$ has more than 1 taxon, then $\{a,b,c,d\}$ is definitely eligible for reduction via Theorem \ref{thm:211eligible} (and can then participate in Reduction 10).
}
\label{fig:211simple}
\end{figure}

\begin{theorem}
\label{thm:211multiedge}
Let $T$ and  $T'$ be two phylogenetic trees on $X$ that cannot be reduced under Reductions 1--7. Let $G$ be the generator underlying a phylogenetic network $N$ on $X$ that displays $T$ and $T'$ such that $r(N)=d_\TBR(T,T')$. Furthermore, let $S =ab|c|d$ be a side of $G$ that is part of a multi-edge $\{u,v\}$. Let $A$, $B$, and $M$ be the three sides incident with $S$ such that $M$ is incident with $u$ and $v$, $A$ is only incident with $u$, and $B$ is only incident with $v$. Then $\{a,b,c,d\}$ is eligible for Reduction 10 if each of the following conditions hold in obtaining $N$ from $G$:
\begin{enumerate}
\item one of $A$ and $B$ is decorated with at least two taxa; and
\item $M$ is decorated with at least one taxon.
\end{enumerate} 
%Assume that Reductions 1-7 have been applied exhaustively. Let $S =ab|c|d$ be a $2|1|1$ side of the underlying generator that \textbf{is} part of a multi-edge. If \emph{\underline{neither}} of the following
%two situations holds, then $\{a,b,c,d\}$ is eligible for reduction via Theorem \ref{thm:211eligible}.
%\begin{enumerate}
%\item there are two distinct sides $A$ and $B$ that are both incident to $S$, but to different ends of $S$, whereby $A$ has at most 1 taxon and $B$ has at most 1 taxon;
%\item the other side in the multi-edge has 0 taxa.
%\end{enumerate}
\end{theorem}
\begin{proof}
Assume without loss of generality that $T$ has cherries $\{a,b\}$ and $\{c,d\}$ and that $T'$ has a pendant 3-chain $(a,b,c)$ with cherry $\{b,c\}$ as illustrated in Figure \ref{fig:211multiedge}. Then the 2-chain $(a,b)$ is CPT eligible. Let $F$ be a maximum agreement forest for $T$ and $T'$ such that $\{a,b\}\subseteq B$. As in the proof of Theorem~\ref{thm:211sparse}, we can assume that, if $\{c\}\notin F$, then $B=\{a,b,c\}$.

Now, for each $Y\in\{A,B,M\}$, let $P_Y$ be the path associated with $Y$ in $N$.
Furthermore, let $I$ be an image of $T$ in $N$, and let $P$ be the path from $a$ to $c$ in $I$.  Since $S$ is a 2-breakpoint side, either $P_A$ and $P_B$ is a subpath of $P$, or $P_M$ is a subpath of $P$. If $P_A$ and $P_B$ are subpaths of $P$ then, since the first condition in the statement of the theorem is satisfied, we can apply the same argument as in the proof of Theorem~\ref{thm:211sparse} to establish that $\{a,b,c,d\}$ is eligible for Reduction 10. We may therefore assume that $P_M$ is a subpath of $P$. As $M$ does not have a breakpoint relative to $T$, $M$ is a 0-breakpoint side or a 1-breakpoint side in which case the breakpoint is relative to $T'$.

Suppose that $M$ is  a 1-breakpoint side $M=|ef$ or $M=ef|$ that is decorated with exactly two taxa $e$ and $f$, or that $M$ is decorated with at least three taxa. Since $M$ has at most one breakpoint, it follows from Observation~\ref{obs:chainseverywhere} that $T$ and $T'$ have a CPT-eligible chain $Z$ whose elements are attached to $M$ in obtaining $N$ from $G$. Applying the same argument as in  the third paragraph of the proof of Theorem~\ref{thm:211sparse} establishes that $\{a,b,c,d\}$ is eligible for Reduction 10. 

Since the second condition in the statement of the theorem holds, we complete the proof by considering two cases depending on whether $M$ is decorated with one or two taxa. For both cases, reconsider $F$ and assume that $B=\{a,b,c\}$. We will see that there exists another maximum agreement forest for $T$ and $T'$ that has the desired property such that $\{a,b,c,d\}$ is eligible for Reduction 10. 

First suppose that $M$ is decorated with only a single taxon $e$. Clearly, $\{d\}$ and $\{e\}$ are elements of $F$. If $M$ is a 0-breakpoint side, let $$F'=(F\setminus \{B,\{d\},\{e\}\})\cup \{\{a,b,d,e\},\{c\}\}$$ be a forest. Since $|F'|<|F|$, it follows from the maximality of $F$ that $F'$ is not an agreement forest for $T$ and $T'$. Hence there exists an element $B'$ in $F'$ (as well as in $F$) such that $T'[B']$ uses the two edges $f$ and $f'$ that are both incident with $p_e$ and not incident with $e$. Let $B_1',B_2'$ be a bipartition of $B'$ such that neither $T'[B_1']$ nor $T'[B_2']$ uses $f$ or $f'$. Then $$F''=(F\setminus \{B,B',\{d\},\{e\}\})\cup \{\{a,b,d,e\},\{c\},B_1',B_2'\}$$ is an agreement forest for $T$ and $T'$ with $|F''|=|F|$ and, hence, $\{a,b,c,d\}$ is eligible for Reduction 10. On the other hand, if $M$ is a 1-breakpoint side, then either $(c,b,a,e)$ or $(d,e)$ is a pendant chain of $T'$. In the former case, let $$F'=(F\setminus \{B,\{e\}\})\cup \{\{a,b,e\},\{c\}\}$$ be a forest and, in the latter case let $$F'=(F\setminus \{B,\{e\},\{d\}\})\cup \{\{a,b\},\{c\},\{d,e\}\}$$ be a forest. Regardless which applies, $F'$ is an agreement forest and, again, $\{a,b,c,d\}$ is eligible for Reduction 10.

Second suppose that $M$ is decorated with exactly two taxa $e$ and $f$ such that the path from $p_a$ to $p_e$ in $T$ does not pass through $p_f$. Clearly, $\{e\}$ and $\{f\}$ are elements of $F$. If $M$ is a 0-breakpoint side, let $$F'=(F\setminus \{B, \{e\},\{f\}\})\cup \{\{a,b,e,f\},\{c\}\}$$ be a forest. As usual, the size of $F'$ contradicts the maximality of $F$. Hence, there exists an element $B'$ in $F'$ (as well as in $F$) such that $T'[B']$ uses an edge $\{p_e,p_f\}$. Let $B_1',B_2'$ be a bipartition of $B'$ such that neither $T'[B_1']$ nor $T'[B_2']$ uses $\{p_e,p_f\}$. Then $$F''=(F\setminus \{B,B',\{e\},\{f\}\})\cup \{\{a,b,e,f\},\{c\},B_1',B_2'\}$$ is an agreement forest for $T$ and $T'$ with $|F''|=|F|$ and, hence, $\{a,b,c,d\}$ is eligible for Reduction 10.

Finally, if $M$ is a 1-breakpoint side with $M=e|f$, let $$F'=(F\setminus \{B,\{e\}\})\cup \{\{a,b,e\},\{c\}\}$$ be a forest. As $F$ is an agreement forest for $T$ and $T'$ and $(c,b,a,e)$ is a chain of $T'$, $F'$ is such a forest as well. Thus, $\{a,b,c,d\}$ is eligible for Reduction 10.\qed

\end{proof} 

\begin{figure}[t]
\center
\scalebox{1}{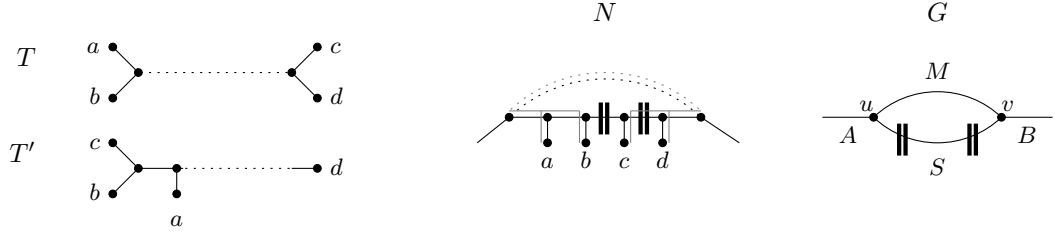}
\caption{The situation described in Theorem \ref{thm:211multiedge}, which concerns $2|1|1$ sides $S=ab|c|d$, where $S$ is part of a multi-edge. The \revision{gray} edges in $N$ indicate an image of $T$ that uses the path in $N$ that is associated with $M$.
%The situation described in Theorem \ref{thm:211multiedge}, which concerns $2|1|1$ sides $ab|c|d$ that are part of a multi-edge. The grey tree is the image of the path from $\{a,b\}$ to $\{c,d\}$ in $T$ in the case that this image passes through $M$, which is the other side in the multi-edge. The high-level idea is that if $M$ has more than 0 taxa, then $\{a,b,c,d\}$ is definitely eligible for reduction via Theorem \ref{thm:211eligible} (and can then participate in Reduction 10).
}
\label{fig:211multiedge}
\end{figure}

%\subsection{Sides with 4 taxa that are loops}

\noindent
Finally, we turn to loops. Consider two phylogenetic trees $T$ and $T'$  and a phylogenetic network $N$ that displays $T$ and $T'$ such that $r(N)=d_\TBR(T,T')$. Let $S$ be a loop side of the generator $G$ that underlies $N$. Then $S$ is decorated with at least one taxon since, otherwise, there exists a phylogenetic network with strictly fewer than $r(N)$ reticulations that displays $T$ and $T'$. Moreover if $A$ denotes the side of $G$ that is incident to $S$ then, because $T$ and $T'$ are connected, $A$ has no breakpoint and $S$ has two breakpoints. Hence every loop is adjacent to a 0-breakpoint side. These observations as well as the next lemma, which shows that loop sides exhibit  clean behavior in terms of being eligible for Reduction 10, will be convenient for the bounding argument of the next section.

%We end this section with the following remark about loop sidesNote also that a loop side with no taxa cannot be part of an optimal generator. Second,
%the side \emph{next} to a loop must be a 0-breakpoint side, otherwise at least one of the trees
%is disconnected (assuming an optimal generator). So a loop cannot be next to a 1-breakpoint or 2-breakpoint side. This is convenient for the bounding argument later, because it means that none of the sides incident to a 4-taxon side can be a loop. 

%The next lemma shows that loop sides exhibit very clean behavior.

\begin{theorem}
\label{thm:211loop}
Let $T$ and  $T'$ be two phylogenetic trees on $X$ that cannot be reduced under Reductions 1--7. Let $G$ be the generator underlying a phylogenetic network $N$ on $X$ that displays $T$ and $T'$ such that $r(N)=d_\TBR(T,T')$. If $S =ab|c|d$ is a side of $G$ that is a loop, then $\{a,b,c,d\}$ is eligible for Reduction 10.
%Assume that Reductions 1-7 have been applied exhaustively. Let $S$ be a loop side of the underlying generator which has 4 taxa $\{a,b,c,d\}$. Then $\{a,b,c,d\}$ is eligible for reduction via Theorem \ref{thm:211eligible}.
\end{theorem}
\begin{proof}
Assume without loss of generality that $T$ has cherries $\{a,b\}$ and $\{c,d\}$ and that $T'$ has a pendant 3-chain $(a,b,c)$ with cherry $\{b,c\}$. Then $(a,b)$ is a CPT-eligible 2-chain. Let $F$ be a maximum agreement forest for $T$ and $T'$ such that $\{a,b\}\subseteq B$ for some element $B$ in $F$. Assume that $\{c\}\notin F$. Then, as before, $c\in B$. Since $T|B=T'|B$, we have $B=\{a,b,c\}$ and, therefore, $\{d\}\in F$. It follows that $$F'=(F\setminus \{B,\{d\}\})\cup \{\{a,b,d\},\{c\}\}$$ is an agreement forest for $T$. Moreover, as $|F'|=|F|$ and $\{c\}\in F'$, $\{a,b,c,d\}$ is eligible for Reduction 10. \qed
%A side that is a loop must have 2 breakpoints, otherwise at least one of the two input trees
%contains a cycle. Hence, the only way for a loop side to have 4 taxa would be if it is of the form
%$S = ab|c|d$.  Assume without loss of generality that $T'$ has
%a pendant 3-chain $(a,b,c)$ where $\{b,c\}$ is the cherry, and $T$ has two cherries $\{a,b\}$ and $\{c,d\}$. 

%We apply the CPT to $\{a,b\}$. If the resulting MAF contains $c$ as a singleton, we are done. Otherwise, assume that $\{a,b,c\}$ is a component of this MAF, which is the only other possibility by the usual
%arguments. In such a MAF, $\{d\}$ must be a singleton. We replace the two components $\{a,b,c\}, \{d\}$ with $\{a,b,d\}, \{c\}$, and we are done.
%\qed
\end{proof}

%\subsubsection{ $2|2$ sides that are not loops or multi-edges}

\section{Putting it all together and bounding the size of the kernel}
\label{sec:alltogether}
In this section, we establish an improved kernel result for computing the TBR distance that is based on Reductions 1--10. We start by bounding the number of certain types of sides in a generator.

\begin{lemma}
\label{l:atmost1}
Let $T$ and $T'$ be two phylogenetic trees on $X$ that cannot be reduced under Reductions 1--10, and let $G$ be a generator that underlies a phylogenetic network $N$ that displays $T$ and $T'$ such that $r(N)=d_\TBR(T,T')$. Furthermore, let $s_2$ be the number of $2|2$ sides of $G$, each being decorated with four taxa that are eligible for Operation P, and let $s_1$ be the number of $1|3$ sides of $G$. Then $s_1+s_2\leq 1$. 
%After exhaustive application of Reductions 1-10,  the number of $1|3$ sides plus the number of  Theorem \ref{thm:22to13} eligible $2|2$ sides in the underlying generator, is at most 1.
\end{lemma}
\begin{proof}
Suppose that $s_1+s_2\geq 2$. By Observation~\ref{obs:13eateachother}, we have $s_1\leq 1$. If $s_1=1$ and $s_2\geq 1$, then $T$ and $T'$ can be reduced by an application of Reduction 9.1. Hence, we may assume that  $s_1=0$ and $s_2\geq 2$. Let $S_1=a'b'|c'd'$ and $S_2=a''b''|c''d''$ be two $2|2$ sides of $G$ such that each of $\{a',b',c',d'\}$ and $\{a'',b'',c'',d''\}$ are both eligible for Operation P. We establish the lemma by showing that we can apply Reduction 9.2, thereby contradicting that $T$ and $T'$ cannot be reduced under Reductions 1--10. Since $\{a'',b'',c'',d''\}$ is eligible for Operation P before this operation is applied to $\{a',b',c',d'\}$, note first that $S_2$ satisfies the conditions in the statement of Theorem~\ref{thm:22sparse} or~\ref{thm:22multiedge}, depending on whether $S_2$ is part of a multi-edge or not.
Now let $S$ and $S'$ be the two phylogenetic trees obtained from $T$ and $T'$, respectively, by applying Operation P to $\{a',b',c',d'\}$. 
%Following the same notation as in the  definition of Operation P, we have $S=T$, where $T$ has a non-pendant chain $(a',b',c',d')$, $T'$ has two cherries $\{a',b'\}$ and $\{c',d'\}$, and $S'$ has chain $(b',c',d')$ with cherry $\{b',c'\}$. 
Crucially, by Corollary~\ref{cor:preserved}, $N$ displays $S$ and $S'$. \revision{As} noted in the corollary the only change is that on side $S_1$ a breakpoint moves slightly. Furthermore, by Theorem~\ref{thm:22to13}, $d_\TBR(T,T')=d_\TBR(S,S')$ which implies that there exists no phylogenetic network that displays $S$ and $S'$ and has strictly fewer than $r(N)$ reticulations. It now follows that $\{a'',b'',c'',d''\}$ is still eligible for Operation P after this operation has been applied to $\{a',b',c',d'\}$ because $S_2$ still satisfies the conditions in the statement of Theorem~\ref{thm:22sparse} or~\ref{thm:22multiedge}, depending on whether $S_2$ is part of a multi-edge or not.\qed
\end{proof}

We next use a pessimistic, but safe, counting argument to finally bound the size of the kernel for computing the TBR distance.\\
\\
\noindent
%\begin{theorem*}
\textbf{Theorem \ref{t:main}.}
\emph{Let $T$ and $T'$ be two phylogenetic trees on $X$ with $d_\TBR(T,T') \geq 2$ that cannot be reduced under Reductions 1--10. Then $|X|\leq 9d_\TBR(T,T')-8$.}
%Reductions 1-10 yield a kernel of size at most $9k-8$ for TBR distance.
%\end{theorem*}
\begin{proof}
Let $N$ be a phylogenetic network that displays $T$ and $T'$ such that $$k=r(N)=d_\TBR(T,T'),$$ and let $G$ be the generator that underlies $N$. By ~\cite[Lemma 1]{tightkernel}, $G$ has $3k-3$ sides \revision{and, by Lemma~\ref{lemma:foundations}, each} side of $G$ is decorated with at most four taxa when obtaining $N$ from $G$. Additionally, by the \revision{latter} lemma, each side that is decorated with four taxa is a $1|3$, $2|2$, or $2|1|1$ side. Moreover, by Lemma~\ref{l:atmost1},  the  number of $2|2$ sides that are eligible for Operation P plus the number of $1|3$ sides is at most one.  We next use the results established in Section~\ref{sec:MAFs-everywhere} to derive the following {\it adjacency rules} for sides of $G$ that are decorated with four taxa.
\begin{enumerate}
\item[A1.] From the contrapositives of Theorems~\ref{thm:22sparse} and~\ref{thm:211sparse}, it follows that each side of $G$ (with possibly one exception by Lemma~\ref{l:atmost1}) that is decorated with four taxa and is not a loop or part of a multi-edge is incident to at least two distinct sides that are each decorated with at most one taxon.
\item[A2.] From the contrapositives of Theorems~\ref{thm:22multiedge} and~\ref{thm:211multiedge}, it follows that each side  of $G$ (with possibly one exception by Lemma~\ref{l:atmost1}) that is decorated with four taxa and is part of a multi-edge is either incident
to at least two distinct sides that are not part of the same multi-edge and each decorated with at most one taxon, or the second side in the multi-edge is  decorated with zero taxa. 
%(Both are possible).
\item[A3.] From the contrapositive of Theorem~\ref{thm:211loop}, it follows that each side of $G$ that is a loop is decorated \sidmablue{with} at most three taxa. Furthermore, \stevenred{to avoid disconnecting $T$ or $T'$}, each such loop side is incident to a 0-breakpoint side that, by Lemma~\ref{lemma:foundations}(b), is decorated with at most three taxa.
%and such a side is incident to a 0-breakpoint side that is decorated with at most three taxa.
\end{enumerate}

%There are $3k-3$ sides in a generator. We know that only sides with $\{0,1,2,3,4\}$ taxa are
%possible after Reductions 1-7 have been applied. Combining all the results from the paper, we know that the following ``adjacency rules'' hold once
%Reductions 1-10 have been applied exhaustively. 
%\begin{enumerate}
%\item[1.] Each side (with possibly one exception: the at most one side mentioned in Lemma \ref{l:atmost1}) that contains 4 taxa and is not part of a multi-edge, is incident to at least
%two distinct sides each with $\leq 1$ taxa.
%\item[2.] Each side (with possibly one exception: the at most one side mentioned in Lemma \ref{l:atmost1}) that contains 4 taxa and \emph{is} part of a multi-edge, is either incident
%to at least two distinct sides each of which has $\leq 1$ taxa, or the second side in the multi-edge
%has 0 taxa. (Both are possible).
%\item[3.] A loop side can have at most 3 taxa, and cannot be incident to a side with 4 taxa.
%\end{enumerate}

We first deal with the exceptional situation that there is a single side $S$ of $G$ that is decorated with four taxa and does not obey A1 or A2.  For the purpose of the upcoming counting argument, we view $S$ as a side that is only decorated with three taxa. This does not affect A1 or A2 because these rules only consider adjacent sides that are decorated with at most one taxon. Furthermore, recalling that $S$ is not a 0-breakpoint side, viewing $S$ as a side that is decorated with three taxa does not affect A3 either. To avoid an underestimate of the final kernel size, we add one to the counting formula below. Next, we consider each  side $S$ of $G$ that is decorated with zero or two taxa and view it in one of the following ways for counting purposes. \stevenred{Note that if $S$ is a loop then by the assumed optimality of \sidmablue{$N$} it must have at least one taxon.}

\begin{enumerate}
\item[1.] If $S$ is not a loop, not part of a multi-edge, and decorated with zero taxa, we view $S$ as a side that is decorated with one taxon. This does not affect A1--A3.
\item[2.] If $S$ is part of a multi-edge and decorated with zero taxa, we view $S$ as a side that is decorated with three taxa and, if subsequently any side $S'$ of $G$ that is incident to $S$ is decorated with four taxa, 
then we view $S'$ as a side that is decorated with three taxa. This cannot decrease the total number of taxa because $S$ is incident to at most three sides that are each decorated with four taxa. Also, we still obey A1--A3, because
any side decorated with four taxa that needed $S$ as a side decorated with zero taxa is now viewed as a side decorated with three taxa.
\item[3.] If $S$ is decorated with  two taxa, we view $S$ as a side that is decorated with three taxa. Again, this does not affect any of A1--A3.
\end{enumerate}

%\begin{enumerate}
%\item[1.] For each side  that is not part of a multiedge and decorated with zero taxa, we increase the number of taxa on that side to 1. (Note
%that this side can still contribute to the $\leq 1$ sides that 4-taxa sides need to be adjacent to).
%\item[2.] For each side in a multi-edge that has 0 taxa, we increase the number of taxa on that
%side to 3, and if there are any 4-taxa sides incident to it, we decrease the number of taxa on each
%of those sides to 3. This cannot decrease the total number of taxa, because at most \emph{three} 4-taxa sides can be incident to a side in a multi-edge. Also, we still obey the adjacency rules, because
%any 4-taxon side that needed the 0-taxon side, has been switched to a 3-taxon side. Hence, we still
%have a valid upper bound.
%\item[3.] Any sides with 2 taxa, we increase to 3. Again, this is safe because a side with 2 taxa
%could not function as a ``$\leq 1$'' side anyway.
%\end{enumerate}

\noindent Now we  still have a valid upper bound on the total number of taxa that decorate sides of $G$, but a simplified counting system because every side is decorated with four, three, or one taxa. Let $p$, $q$, and $r$ be the number of sides of $G$ that are decorated with  with four, three, and one taxa respectively.
%At this stage we still have a valid upper bound on the size of the kernel, but a simplified counting
%system: each side has 4, 3 or 1 taxa.\\
%Let $p, q, r$ be the number of sides with 4, 3, 1 taxa respectively.\\
%\\
Then we have the following optimization problem, where the $+1$ in the objective function is due to the possibly undercounted
side decorated with four taxa that is mentioned above and that we view as a side decorated with three taxa.

\begin{verbatim}
Maximize 4p + 3q + 1r + 1
subject to
p + q + r = 3k - 3
p <= 2k
r >= (2/4)p and
p, r, q >= 0 (and integer)
\end{verbatim}

\noindent The $p\leq 2k$ inequality occurs because, by Lemma~\ref{lemma:foundations},
a side that is decorated with four taxa has at least one breakpoints and there are $2k$ breakpoints in total (i.e., $k$ breakpoints for each tree). Furthermore, each side that is decorated with one  taxon  can be incident to at most four sides that are each decorated with four taxa. On the other hand, since $T$ and $T'$ cannot be further reduced under any of Reductions 1--10, each side that is decorated with four taxa needs to be incident to at least two sides decorated with one taxon. This implies that 
$r \geq  (2/4)p$.
We next substitute $q=(3k-3)-p-r$ and this gives

\begin{verbatim}
Maximize 9k + p - 2r - 8
subject to
p <= 2k
r >= (1/2)p and
p, r, q >= 0 (and integer).
\end{verbatim}

\noindent The fact that  $r \geq (1/2)p$ implies that the term $(p-2r)$ in the objective function is at most 0. We conclude that $|X|\leq 9k-8=9d_\TBR(T,T')-8$ is an upper bound on the size of our kernel.
\qed
\end{proof}

The bound $9k-8$ is tight up to an additive term of 1, as the following theorem shows. The additive term is due, in the above analysis, to the at most one generator side with four taxa that does not obey A1 or A2. \revision{To establish the next theorem, we need the following definitions.} A {\it binary character} $f$ on $X$ is a function that assigns each element in $X$ to an element in $\{0,1\}$. Let $T$ be an \revision{phylogenetic tree on $X$} with vertex set $V$. An {\it extension} $g$ of $f$ to $V$ is  a function $g$ that assigns each element in $V$ to an element in $\{0,1\}$ such that $g(x)=f(x)$ for each $x\in X$. The {\it parsimony score} of $f$ on $T$, denoted by $l_f(T)$, denotes the minimum number of edges $\{u,v\}$ in $T$ such that $g(u)\ne g(v)$, ranging over all extensions of $f$. Now, for two  \revision{phylogenetic trees $T$ and $T'$ on $X$}, the {\it maximum parsimony distance on binary characters} $d^2_{\MP}$ is defined as  $d^2_{\MP}(T,T')=\max_f |l_f(T)-l_f(T')|$ where $f$ ranges over all binary characters on $X$. It is well-known that $d_\TBR(T,T') \geq d^2_{\MP}(T,T')$ \cite{fischer2014}. 

\begin{theorem}
\label{thm:tight}
For each $k \geq 3$ there exist two phylogenetic trees $T_k$ and  $T'_k$ with $9k-9$ taxa and $d_\TBR(T_k,T'_k) = k$ that cannot be reduced under Reductions 1--10.
\end{theorem}
\begin{proof}
Let $k \geq 3$. We proceed by building a specific ladder-like generator $G_k$, converting this to a phylogenetic network $N_k$, and extracting the trees $T_k$ and $T'_k$ from this. We will then
prove that $d_\TBR(T_k,T'_k) = k$ and that \revision{$T_k$ and $T_k'$} are irreducible under Reductions 1--10.

Generator $G_k$ is built as follows. We take the rectangular $2 \times (k+1)$ grid on $2(k+1)$ \revision{vertices} and suppress the four corner vertices of degree 2. This creates a cubic multigraph $G_k$ with $3(k-1)$ sides. Note that $G_k$ has exactly two pairs of multi-edges. We create $N_k$ by decorating each side of $G_k$ with 3 taxa. Let $X_k$ be the set of all taxa added; we have  $|X_k| = 9k-9$. By construction, $r(N_k) = k$. See Figure \ref{fig:tight} for the situation $k=5$. Let $T_k$ (respectively, $T'_k$) be the tree displayed by $N_k$ that is induced by the $k$ solid (respectively, hollow) breakpoints as indicated in the figure. Given that $r(N_k)=k$ we have $d_\TBR(T_k,T'_k) \leq k$. 

To prove that $d_\TBR(T_k,T'_k) \geq k$ we
use the same lower-bounding technique as \cite{tightkernel,kelk2020new}. \revision{To this end,} 
%a {\it binary character} $f$ on $X$ is a function that assigns each element in $X$ to an element in $\{0,1\}$. Let $T$ be an \revision{phylogenetic tree on $X$} with vertex set $V$. An {\it extension} $g$ of $f$ to $V$ is  a function $g$ that assigns each element in $V$ to an element in $\{0,1\}$ such that $g(x)=f(x)$ for each $x\in X$. The {\it parsimony score} of $f$ on $T$, denoted by $l_f(T)$, denotes the minimum number of edges $\{u,v\}$ in $T$ such that $g(u)\ne g(v)$, ranging over all extensions of $f$. Now, for two  \revision{phylogenetic trees $T$ and $T'$ on $X$}, the {\it maximum parsimony distance on binary characters} $d^2_{\MP}$ is defined as  $d^2_{\MP}(T,T')=\max_f |l_f(T)-l_f(T')|$ where $f$ ranges over all binary characters on $X$. It is well-known that $d_\TBR(T,T') \geq d^2_{\MP}(T,T')$ \cite{fischer2014}. Specifically to prove that $d_\TBR(T_k,T'_k) \geq k$ 
it is sufficient to give a binary character $f$ on $X_k$ such that 
$|l_f(T_k)-l_f(T'_k)| \geq k$. We define $f$ by assigning 0 to each taxon to the left of the \revision{gray} line, as indicated in Figure~\ref{fig:tight}, and assigning 1 to all other taxa, i.e. those to the right of the \revision{gray} line. It
is easy to check that $l_f(T_k)=1$ and that, by Fitch's algorithm or similar \cite{fitch1971}, $l_f(T'_k) \geq (k+1)$, so $d_\TBR(T_k,T'_k)  \geq |l_f(T_k)-l_f(T'_k)| \geq k$ as required. This concludes the proof that
$d_\TBR(T_k,T'_k) = k$.

Regarding irreducibility, it is helpful to first inventarise some topological features of $T_k$ and $T'_k$. They have no common pendant subtrees of size 2 or larger, so Reduction 1 is excluded, and they have no common chains of length 4 or longer, so Reduction 2 is excluded. Crucially, each tree has exactly
one pendant 3-chain but this is \emph{not} common with the other tree. For $T_k$ this is $(p,q,r)$, where $\{q,r\}$ is its cherry, and for $T'_k$ this is $(s,t,u)$, where $\{t,u\}$ is its cherry. Hence, Reductions 3, 4, 6 and 7 are excluded.  Recalling the definition of Reduction 5, we see that if the preconditions for this reduction rule hold, then it also follows that $(\ell_1, \ell_2, x)$ is a pendant 3-chain in one tree (with $\{\ell_2, x\}$ the cherry) and $(\ell_1, \ell_2)$ is a 2-chain common to both trees. %and $(x, \ell_4, \ell_3)$ is a pendant 3-chain in the other (with $\{\ell_4,x\}$ the cherry)
 As noted already each tree has exactly one pendant 3-chain. Without loss of generality (due to symmetry between $T_k$ and $T'_k$), observe that the single pendant 3-chain  $(p, q, r)$ in $T_k$, where $\{q, r\}$ is the cherry, has the property that $(p, q)$ is \emph{not} a 2-chain in $T'_k$, so Reduction 5 cannot apply. 

We now turn to the new reduction rules. Consider Reduction 8. This is built on Reduction 8A, which requires a common 3-chain that is pendant in one tree: again, this does not exist, so the reduction is
excluded. The same fact immediately excludes Reduction 9.1. Reduction 9.2 requires Operation P to execute, and this operation requires one of the trees to have a non-pendant 4-chain $(a,b,c,d)$ and the other tree to have cherries $\{a,b\}$ and $\{c,d\}$. Each of $T_k$ and $T'_k$ has exactly $(k+1)$ cherries, but no pair of these cherries combine to form a 4-chain in the other tree, so Operation P cannot apply. (Viewed
from the contrapositive perspective: any 4-chain must contain at least one taxon that is not in a cherry in the other tree). Finally, consider Reduction 10. The preconditions here require one \revision{of $T_k$ and $T_k'$} to have a pendant \mbox{3-chain} $(a,b,c)$ where $\{b,c\}$ is the cherry, and the other tree to have the cherry $\{a,b\}$. However, the single pendant 3-chain in (without loss of generality) $T_k$, $(p,q,r)$ where $\{q,r\}$ is the cherry has the property that the first taxon $p$ on the chain is definitely not in a cherry in the other tree, so Reduction 10 cannot execute. We are done. \qed
\end{proof}

\begin{figure}[t]
\center
\scalebox{1.5}{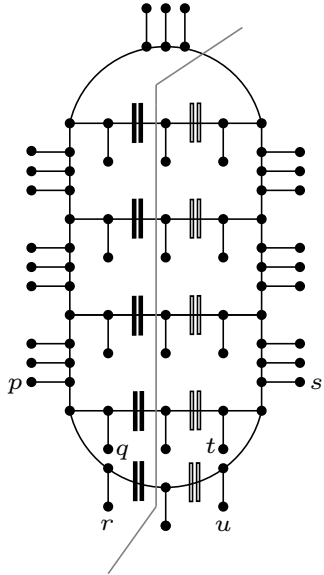}
\caption{The phylogenetic network $N_k$ for $k=5$, as constructed in the proof of Theorem \ref{thm:tight}. The $k$ solid double bars represent the breakpoints for the first tree $T_k$ and the $k$ hollow double bars
represent the breakpoints for the second tree $T'_k$. Taxa $p$, $q$, and $r$ form the unique pendant 3-chain of $T_k$ and $s$, $t$, and $u$ form the unique pendant 3-chain of $T'_k$. The \revision{gray} line is used in the proof \revision{to show} that
$d_\TBR(T_k,T'_k) \geq k$.}
\label{fig:tight}
\end{figure}

%\begin{theorem}
%Reductions 1-10 yield a kernel of size at most $9k-8$ for TBR distance.
%\end{theorem}

\section{Conclusion and future work}

There are a number of interesting future research directions. The most obvious direction is to design new reduction
rules capable of further improving the current $9k-8$ bound. How far
below $9k-8$ can we go, and what is the trade off between proof complexity and the obtained decrease in kernel size?
Specifically, the existing reduction rules and
their associated proofs are already rather complex, requiring extensive auxiliary mathematical machinery and
quite some case-checking. It is natural to ask whether the design of further rules and the related proofs can be
streamlined and simplified in some fashion by deepening our understanding of the combinatorial behaviour of agreement forests.
In the meantime (semi-)automated tools for proof verification could be utilized to help keep case-checking under control. We note also that
Reduction 8 hints at a wider family of reduction rules. Essentially, it gives us a general recipe for moving certain edges
around in the trees such that $d_{\TBR}$ is preserved: this allows us to rearrange the trees in such a way that \emph{other} reduction rules are
triggered. We expect that such `indirect' reduction rules will be very useful in the future. \stevenred{Also, the fact that Reduction 8 actually undoes an application of the chain reduction
 \revision{is interesting}: it takes a step `back', in order to move forward. As discussed in \cite{figiel2022there} this phenomenon merits further study.}

Next, an empirical study
in the spirit of \cite{van2022reflections} could investigate how much extra reductive power the new $9k-8$ rules have in practice;
the rules for the $11k-9$ kernel do have more practical effect than the $15k-9$ rules, does this trend continue?
Another angle to explore is to translate the new reduction rules onto other agreement-forest based phylogenetic distances to obtain smaller
kernels there; this has already been effective in designing new reduction rules for Rooted Subtree Prune and Regraft distance \cite{kelk2022cyclic}.

\section{Acknowledgements}

Steven Kelk and Simone Linz were supported by the New Zealand Marsden Fund. Ruben Meuwese was supported by the Dutch Research Council (NWO) KLEIN 1 grant \emph{Deep kernelization for phylogenetic discordance}, project number OCENW.KLEIN.305. \stevenred{We thank Steve Chaplick for useful discussions.}

\bibliography{article}{}
\bibliographystyle{plain}

\newpage
\appendix

\section*{Appendix}

\section{Proof of Theorem~\ref{t:lovely-chains}}

\noindent{\it Proof of Theorem~\ref{t:lovely-chains}.}  
Let $e=\{u,v\}$ be the interrupter in $T'$ of $C$, and let $v$ be the unique common neighbor of $p_b$ and $p_c$ in $T'$. Towards  a contradiction assume that the result does not hold. Let $F^*$ be a maximum agreement forest for $T$ and $T'$. Then there exists an element $B\in F^*$ such that $T'[B]$ uses $e$. Let $Q,R$ be the bipartition of $X\setminus \{a,b,c,d\}$ such that, in $T$, the path from each element in $Q$ to $a$ is shorter than its path to $d$, and the path from each element in $R$ to $d$ is shorter than its path to $a$. Similarly, let $Q',R',S'$ be the tripartition of $X\setminus \{a,b,c,d\}$ such that, in $T'$, the path from each element in $Q'$ to $a$ is shorter than its path to $d$, the path from each element in $R'$ to $d$ is shorter than its path to $a$, and the path from each element in $S'$ to $a$ has the same length than its path to $d$. This setup is illustrated in Figure~\ref{fig:setup}. We next define five sets that will be useful throughout the proof. Specifically, let $B_Q=B\cap Q$, $B_R=B\cap R$, $B_{Q'}=B\cap Q'$, $B_{R'}=L\cap R'$, and $B_{S'}=B\cap S'$. As $T'[B]$ uses $e$, note that $B_{S'}$ is non empty. Moreover, since $T|B= T'|B$, it follows that $|B\cap C|\leq 3$. We freely use the previous two properties of  $B_{S'}$ and $B\cap C$, respectively, throughout the remainder of the proof. To establish the result, we next consider three cases that each have several subcases. In all (sub)cases we will show that there exists a maximum agreement forest $F$ that does not use $e$.\\
\\
\noindent {{\bf Case 1.} $B_{Q'}=\emptyset$, and $B_{R'}=\emptyset$} \\
Since $T'[B]$ uses $e$, we have $1\leq |B\cap C|\leq 3$. Furthermore, there is no element in  $F^*\setminus \{B\}$ that contains an element of $Q'$ and an element of $R'$. However, if $B_{S'}\subseteq Q$ or $B_{S'}\subseteq R$, then there can be an element in $F^*\setminus B$ that has a non-empty intersection with $C$ and a non-empty intersection with one of $Q$ and $R$. There are three subcases to consider for Case~1.

First suppose that $B\cap \{a,b\}\ne\emptyset$ and $B\cap \{c,d\}\ne\emptyset$.
If $|B\cap C|=2$, then $B_{S'}\subseteq Q$ or $B_{S'}\subseteq R$. On the other hand, if $|B\cap C|=3$ then, because $T[B]=T'[B]$, we have $B_{S'}\subseteq R$ when $\{a,b\}\subset B$, and $B_{S'}\subseteq Q$ when  $\{c,d\}\subset B$.
Considering $T[B]$, it follows that there exists an element $\ell\in C\setminus B$ such that $\{\ell\}\in F^*$. Hence $$F=(F^*\setminus \{
B,\{\ell\}\})\cup\{B_{S'},(B\setminus B_{S'})\cup\{\ell\}\}$$ is an agreement forest for $T$ and $T'$ with $|F|=|F^*|$.

Second suppose that $B\cap \{a,b\}=\emptyset$ and $|B\cap \{c,d\}|=2$.
Then again $B_{S'}\subseteq Q$ or $B_{S'}\subseteq R$.  Let $B'$ be the element in $F^*$ that contains $b$. Note that $|B'|\geq 1$ if $B_{S'}\subseteq R$ and $|B'|=1$ if $B_{S'}\subseteq Q$. Moreover if $B'$ contains an element in $X\setminus C$, then $B'\setminus C\subseteq Q$ and  $B'\setminus C\subseteq Q'$.  Hence $$F=(F^*\setminus \{B,B'\})\cup\{B_{S'},(B\setminus B_{S'})\cup B'\}$$ is an agreement forest for $T$ and $T'$ with $|F|=|F^*|$. An analogous symmetric analysis applies when $|B\cap \{a,b\}|=2$ and $B\cap \{c,d\}=\emptyset$.

Third suppose that $B\cap \{a,b\}=\emptyset$ and $|B\cap \{c,d\}|=1$.
If $\{c,d\}\cap B=\{c\}$, let $B'$ be the element in $F^*$ that contains $b$, and if $\{c,d\}\cap B=\{d\}$, let $B'=\{c\}$. In the latter case, note that $B'\in F^*$ because of $T'[B]$. Now, under the assumption that $B_Q=\emptyset$ or $B_R=\emptyset$, it follows that $$F=(F^*\setminus \{B,B'\})\cup\{B_{S'},(B\setminus B_{S'})\cup B'\}$$ is an agreement forest for $T$ and $T'$ with $|F|=|F^*|$. We may therefore assume that $B_Q\ne\emptyset$ and $B_R\ne \emptyset$, that is $B_{S'}=B_Q\cup B_r$. Then, because of $T[B]$, there are three singletons $\{\ell\},\{\ell'\}$, and $\{\ell''\}$ in $F^*$ such that $\{\ell,\ell',\ell''\}=C\setminus B$. In other words, there is no element in $F^*\setminus \{B\}$ that contains an element in $C$ and an element not in $C$. Hence $$F=(F^*\setminus \{B,\{\ell\},\{\ell'\},\{\ell''\}\})\cup\{B_Q,B_R,(B\setminus B_{S'})\cup\{\ell,\ell',\ell''\}\}$$ is an agreement forest for $T$ and $T'$ with $|F|<|F^*|$. An analogous symmetric analysis applies when $|B\cap \{a,b\}|=1$ and $B\cap \{c,d\}=\emptyset$.
\\

\noindent {{\bf Case 2.} $B_I=\emptyset$ and $B_J\ne\emptyset$ with $\{I,J\}=\{Q',R'\}$}\\
Without loss of generality, we may assume that $B_{Q'}=\emptyset$ and $B_{R'}\ne\emptyset$. 
%Observe that each element in  $F^*\setminus\{B\}$ has a non-empty intersection with at most one of $Q'$, $R'$, and $S'$. 
There are four subcases to consider for Case 2.

First suppose that $B\cap \{a,b\}\ne\emptyset$ and $B\cap \{c,d\}\ne\emptyset$. Since $T|B= T'|B$, it follows that $B_{S'}=B_Q$ and $B_{R'}=B_R$. Moreover, if $\{a,b\}\subset B$, then $T|B\ne T'|B$, which implies that there exists an element $\ell\in \{a,b\}$ such that $\{\ell\}\in F^*$. Hence $$F=(F^*\setminus \{B,\{\ell\}\})\cup\{B_{S'},(B\setminus B_{S'})\cup\{\ell\}\}$$ is an agreement forest for $T$ and $T'$ with $|F|=|F^*|$.

Second suppose that $B\cap \{a,b\}\ne\emptyset$ and $B\cap \{c,d\}=\emptyset$.
Then clearly $\{c\},\{d\}\in F^*$. Since $B\cap C\ne\emptyset$, there exists no element in $F^*\setminus\{B\}$ that contains an element of $Q$ and an element of $R$.
% and there is no element in $F^*\setminus\{B\}$ that contains an element of $C$ and an element of $Q\cup R$. 
Moreover, if $B\cap C\in \{\{a\}, \{a,b\}\}$, then no element in $F^*\setminus \{B\}$ contains an element in $C$ and an element in $X\setminus C$. Hence, 
$$F=(F^*\setminus \{B,\{c\},\{d\}\})\cup\{B_{S'},B_{R'},(B\cap C)\cup\{c,d\}\}$$
is an agreement forest for $T$ and $T'$ with $|F|=|F^*|$. For the remainder of this subcase, assume that $B\cap \{a,b\}=\{b\}$. If $B_Q\ne\emptyset$ then $\{a\}\in F^*$ and $F$ is again an agreement forest for $T$ and $T'$. Lastly, if $B_Q=\emptyset$, let $B'$ be the element in $F^*\setminus \{B\}$ such that $a\in B'$. As $T'[B]$ and $T'[B']$ are vertex disjoint, we have $B'\setminus \{a\}\subseteq Q'$. It is now straightforward to check that $F$ is an agreement forest for $T$ and $T'$.

Third suppose that $B\cap \{a,b\}=\emptyset$ and $B\cap \{c,d\}\ne\emptyset$.
If $B\cap \{c,d\}=\{c,d\}$, then $B_{S'}=B_Q$ and $B_{R'}=B_R$. It follows that $\{a\},\{b\}\in F^*$. On the other hand, if $B\cap \{c,d\}=\{c\}$ (resp. $B\cap \{c,d\}=\{d\}$), then $\{d\}\in F^*$ (resp. $\{c\}\in F^*$). Hence, $$F=(F^*\setminus \{B,\{\ell\}\})\cup\{B_{S'},(B\setminus B_{S'})\cup\{\ell\}\}$$ is an agreement forest for $T$ and $T'$ with $|F|=|F^*|$ and where $\ell\in\{b,c,d\}$ depending on which of the three elements is a singleton in $F^*$.

Fourth suppose that $B\cap C=\emptyset$. Clearly, $\{c\}, \{d\}\in F^*$. If there exists no $B'\in F^*$ such that $B'\cap (Q\cup\{a,b\})\ne\emptyset$ and $B'\cap R\ne\emptyset$, then $$F=(F^*\setminus\{B,\{c\},\{d\}\})\cup\{B_{S'},B\setminus B_{S'},\{c,d\}\}$$ is an agreement forest for $T$ and $T'$ with $|F|=|F^*|$. Hence, we may assume that $B'$ exists. Furthermore, assume first that $B'\ne B$. Since $T|B'= T'|B'$, one of the following properties holds depending on which of $a$ and $b$ is contained in $B'$.
\begin{enumerate}
\item If $B'\cap\{a,b\}=\emptyset$, then $\{a\},\{b\}\in F^*$ because $T[B']$ uses $\{p_a,p_b\}$.
\item If $B'\cap\{a,b\}=\{a\}$, then $\{b\}\in F^*$ because $T[B']$ uses $\{p_a,p_b\}$.
\item If $B'\cap\{a,b\}=\{b\}$, then $\{a\}\in F^*$ because $T'[B]$ and $T'[B']$ are vertex disjoint.
\item If $B'\cap\{a,b\}=\{a,b\}$, then $B'\cap Q=\emptyset$ because $T'[B]$ and $T'[B']$ are vertex disjoint.
\end{enumerate}
It now follows that $$F=(F^*\setminus\{B,B',\{a\},\{b\},\{c\},\{d\}\})\cup\{B_{S'},B\setminus B_{S'},C,B'\cap Q, B'\cap R\}$$ is an agreement forest for $T$ and $T'$ with $|F|< |F^*|$ if Property (1) applies, $$F=(F^*\setminus\{B,B',\{\ell\},\{c\},\{d\}\})\cup\{B_{S'},B\setminus B_{S'},C,B'\cap Q, B'\cap R\}$$ is an agreement forest for $T$ and $T'$ with $|F|\leq |F^*|$ and $\ell=b$ (resp. $\ell=a$) if Property (2) (resp. Property (3)) applies, and $$F=(F^*\setminus\{B,B',\{c\},\{d\}\})\cup\{B_{S'},B\setminus B_{S'},C, B'\cap R\}$$  is an agreement forest for $T$ and $T'$ with $|F|= |F^*|$ if Property (4) applies. Now assume that $B'=B$. Clearly $\{a\},\{b\}\in F^*$. 
Consider the bipartition $B_{S'},B_{R'}$ of $B$. 
%Note that one of $T[B_{S'}]$ and $T[B_{R'}]$ uses the edge $\{p_c,p_d\}$, which exists since, by the choice of $B'$, the vertices $p_c\ne p_d$ in $T$. 
Let $B_{S'}^Q=B_{S'}\cap Q$, $B_{S'}^R=B_{S'}\cap R$, $B_{R'}^Q=B_{R'}\cap Q$, and $B_{R'}^R=B_{R'}\cap R$. It now follows that $$F=(F^*\setminus \{B,\{a\},\{b\},\{c\},\{d\}\})\cup\{B_{S'}^Q,B_{S'}^R,B_{R'}^Q,B_{R'}^R,C\}$$ is an agreement forest for $T$ and $T'$. In particular, since $T|B= T'|B$, at least one element in $\{B_{S'}^Q,B_{S'}^R,B_{R'}^Q,B_{R'}^R\}$ is the empty set and, so, $|F|<|F^*|$. \\
\\
\noindent {{\bf Case 3.} $B_{Q'}\ne\emptyset$ and $B_{R'}\ne\emptyset$}\\
Since $T|B= T'|B$, we have that $B\cap\{a,b\}=\emptyset$ or $B\cap\{c,d\}=\emptyset$. Hence $|B\cap C|\leq 2$. There are three subcases to consider for Case 3.

First suppose that $|B\cap C|=2$.
Then there exist  two distinct element $\ell,\ell'\in C$ such that $\{\ell\},\{\ell'\}\in F^*$. If $B\cap C=\{a,b\}$ (resp. $B\cap C=\{c,d\}$), then  $B_{S'}\cup B_{R'}\subseteq R$ and $B_{Q'}\subseteq Q$ (resp. $B_{S'}\cup B_{Q'}\subseteq Q$  and $B_{R'}\subseteq R$). Hence, $$F=(F^*\setminus\{B,\{\ell\},\{\ell'\}\})\cup\{B_{S'},(B\setminus B_{S'})\cup\{\ell,\ell'\}\}$$ is an agreement forest for $T$ and $T'$ with $|F|< |F^*|$.

Second suppose that $|B\cap C|= 1$. Then there exist three distinct element $\ell,\ell',\ell''\in C$ such that $\{\ell\},\{\ell'\},\{\ell''\}\in F^*$. In turn, because there is no element in $F^*\setminus B$ that contains an element in $Q$ and an element in $R$, this implies that, except for possibly $B$, no other element in $F^*$ uses any of the three edges  $\{p_a,p_b\}$, $\{p_b,p_c\}$, and $\{p_c,p_d\}$.
%Moreover, as $T|B=T'|B$, it follows $B_Q\ne\emptyset$ and $B_R\ne\emptyset$ which implies that $T[B]$ uses the edges $\{p_a,p_b\}$, $\{p_b,p_c\}$, and $\{p_c,p_d\}$. 
Lastly, since $B\cap C\ne\emptyset$, $B_{S'}$  is either contained in $Q$ or $R$.
It now follows that $$F=(F^*\setminus\{B,\{\ell\},\{\ell'\},\{\ell''\}\})\cup\{B_{Q'},B_{R'},B_{S'}, C\}$$ is an agreement forest for $T$ and $T'$ with $|F|=|F^*|$.

Third suppose that $|B \cap C|=0$.
Then each element in $C$ is a singleton in $F^*$. If there exists no element $B'\in F^*$ such that $T[B']$ uses an edge $\{p_\ell,p_{\ell'}\}$ for two distinct $\ell,\ell'\in C$,
%This edge exists since $p_a$, $p_b$, $p_c$, and $p_d$ are pairwise distinct vertices in $T'$. 
then $$F=(F^*\setminus\{B,\{a\},\{b\},\{c\},\{d\}\})\cup\{B_{Q'},B_{R'},B_{S'},C\}$$ is an agreement forest for $T$ and $T'$ with $|F|< |F^*|$. Otherwise, if $B'$ exists, then $B'$ is the unique such element since $B'\cap Q\ne\emptyset$ and $B'\cap R\ne\emptyset$. Hence, assuming that $B'\ne B$,
%either $B\subseteq Q$ or $B\subseteq R$. Hence 
$$F=(F^*\setminus\{B,B',\{a\},\{b\},\{c\},\{d\}\})\cup\{B_{Q'},B_{R'},B_{S'},B'\cap Q, B'\cap R,C\}$$ is an agreement forest for $T$ and $T'$ with $|F|= |F^*|$. Lastly, if $B'=B$, consider the three sets $B_{Q'}$, $B_{R'}$, and $B_{S'}$. Since $T|B=T'|B$, at most one of these three sets, say $B_{S'}$, has a non-empty intersection with $Q$ and a non-empty intersection with $R$. Then $$F=(F^*\setminus\{B,\{a\},\{b\},\{c\},\{d\}\})\cup\{B_{Q'},B_{R'},B_{S'}\cap Q, B_{S'}\cap R,C\}$$ is an agreement forest for $T$ and $T'$ with $|F|= |F^*|$. An analogous argument holds if $B_{Q'}$ or $B_{R'}$ has a non-empty intersection with $Q$ and a non-empty intersection with $R$.
\qed 

\section{Explicit descriptions of Algorithm 1 and 2 for testing eligibility for Operation P or Reduction 10.}

Reductions 9 and 10 rely on Algorithms 1 and 2 to test eligibility. These algorithms do not have access to the underlying generator and have to search for the corresponding
structures in two  phylogenetic trees.  The algorithms closely mirror the analyses in the proofs of Theorems~\ref{thm:22sparse}--\ref{thm:211loop}. 

%Theorems \ref{thm:22to13} and \ref{thm:211eligible}, respectively.

\subsection{Algorithm 1 tests whether $\{a,b,c,d\}$ is eligible for Operation P}

Assume that $T'$ has
cherries $\{a,b\}$ and $\{c,d\}$ and $T$ has a non-pendant chain $(a,b,c,d)$ where $a$ and $d$
are the outermost leaves on the chain. This can easily be confirmed in polynomial time.

If at least one of the following polynomial-time checkable conditions is true, return YES i.e. Operation P can be applied. If none of them are true, return NO/DON'T KNOW.\footnote{We write NO/DON'T KNOW because it might still be possible that $\{a,b,c,d\}$ is eligible for Operation P but for reasons that fall outside the conditions described in Theorems~\ref{thm:22sparse} and \ref{thm:22multiedge} (which
are those checked by the algorithm). However,
we do not care about such cases. Functionally speaking a NO/DON'T KNOW answer is therefore interpreted simply as NO. The same comment holds for Algorithm 2.}

\begin{enumerate}
\item In $T'$, the path from $p_a$ to $p_c$ passes through at least one CPT-eligible chain $Z$ where $Z \cap \{a,b,c,d\} = \emptyset$.
\item Any of situations (a)--(g) from Figure \ref{fig:algorithm1} occur. 
\end{enumerate}

Step 1 captures the parts of Theorems \ref{thm:22sparse} and \ref{thm:22multiedge} when the path $P$, passing through side $A$, $B$ or $M$ depending on the situation, contains at least one CPT-eligible chain. (It does not matter if the chain found does not actually lie on $A$, $B$ or $M$: it is still correct in this case to conclude that Operation P is eligible). Situation (a) of Step 2 covers the situation in Theorem \ref{thm:22sparse} when the side $A$ has
0 breakpoints and two taxa $e$ and $f$. Situation (b) is when side $A$ has the form $e|f$. Situations (c) and (d) are symmetrical to (a) and (b): when the path $P$ uses
side $B$ rather than $A$. Situations (e)--(g) concern the cases in Theorem \ref{thm:22multiedge} where $M$ is a side $e|f$, $e|$ or $|e$ respectively (and the breakpoint is with respect to $T$).

\subsection{Algorithm 2 tests whether $\{a,b,c,d\}$ is eligible for Reduction 10}

Assume without loss of generality that $T'$ has
a pendant 3-chain $(a,b,c)$ where $\{b,c\}$ is the cherry, and $T$ has two cherries $\{a,b\}$ and $\{c,d\}$. This can easily be confirmed in polynomial time.

If at least one of the following polynomial-time checkable conditions is true, return YES i.e. Reduction 10 can be applied. If none of them are true, return NO/DON'T KNOW.

\begin{enumerate}
\item In $T$, the path from $p_a$ to $p_c$ passes through at least one CPT-eligible chain  $Z$ where $Z \cap \{a,b,c,d\} = \emptyset$.
\item Any of situations (a)--(j) from Figure \ref{fig:algorithm2} occur. 
\end{enumerate}

Step 1 captures the parts of Theorems \ref{thm:211sparse} and \ref{thm:211multiedge} when the path $P$, passing through side $A$, $B$ or $M$ depending on the situation, contains at least one CPT-eligible chain. (It does not matter if the chain found does not actually lie on $A$, $B$ or $M$: it is still correct in this case to conclude that Reduction 10 is eligible). Situation (a) of Step 2 covers the situation in Theorem
\ref{thm:211sparse} when side $A$ is a 0-breakpoint side with two taxa $e$ and $f$, and (b) when $A$ is a side $e|f$. Situation (c) covers the situation when side $B$ is a 0-breakpoint side with
taxa $e$ and $f$, and situation (d) when side $B$ is a $e|f$ side; note that situations (c) and (d) are not entirely symmetrical to situations (a) and (b) due to the inherent asymmetry of $2|1|1$ sides. Situations (e)--(i) concern Theorem \ref{thm:211multiedge}. In particular, (e) is when $M$ is a 0-breakpoint side with two taxa $e$ and $f$ and
(f) is when $M$ is a 1-breakpoint side $e|f$ (where the breakpoint is with respect to $T'$). Situation (g) is when $M$ is a 0-breakpoint side with only one taxon $e$, (h) is when $M$ is a side $e|$, and
(i) is when $M$ is a side $|e$. Again, the breakpoints here are with respect to $T'$. Situation (j) reflects Theorem \ref{thm:211loop}.

\begin{figure}[t]
\center
\scalebox{1}{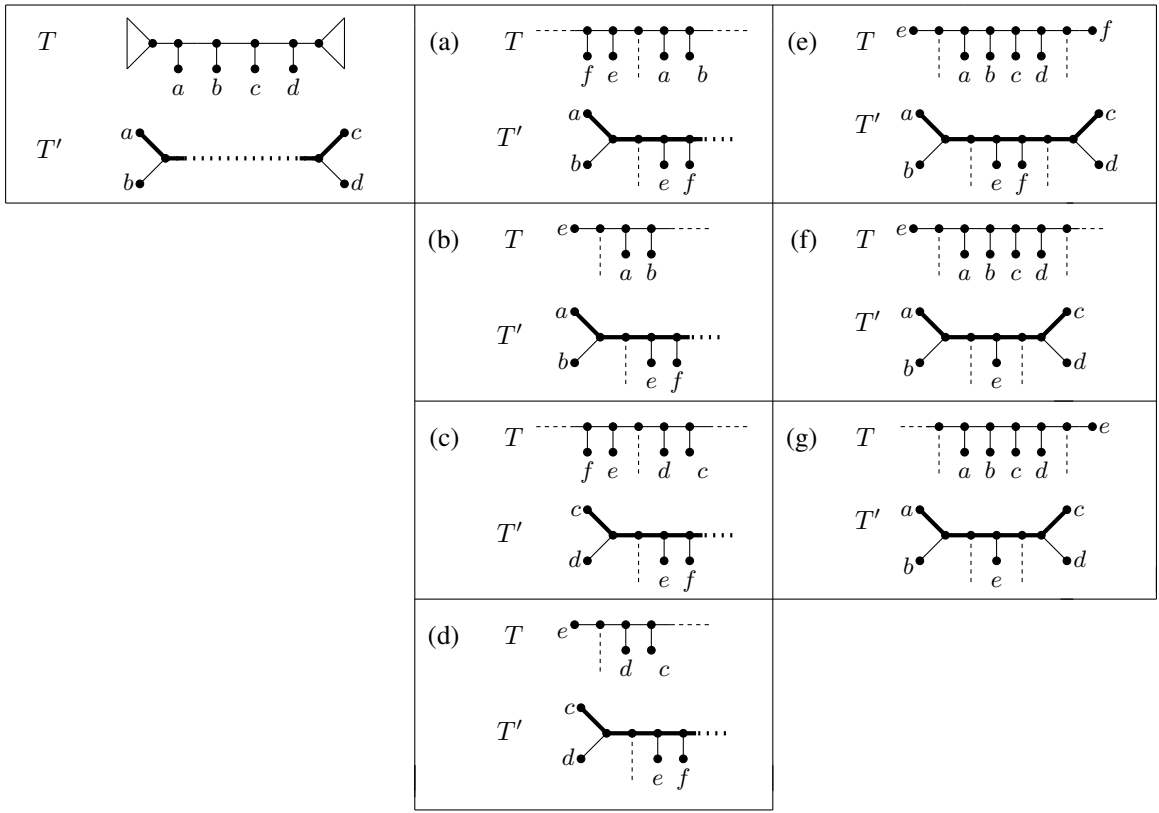}
\caption{Tree topologies checked by Algorithm 1. Path $P$, as used in Theorems \ref{thm:22sparse} and \ref{thm:22multiedge}, is indicated in bold. Solid lines are edges. Dotted and dashed lines are subtrees that can be optionally present
in the tree. Figures (a)--(d) correspond to the situation when the $2|2$ side in the underlying generator is a simple edge, and (e)--(g) to when it is a multi-edge.}
\label{fig:algorithm1}
\end{figure}

\begin{figure}[t]
\center
\scalebox{1}{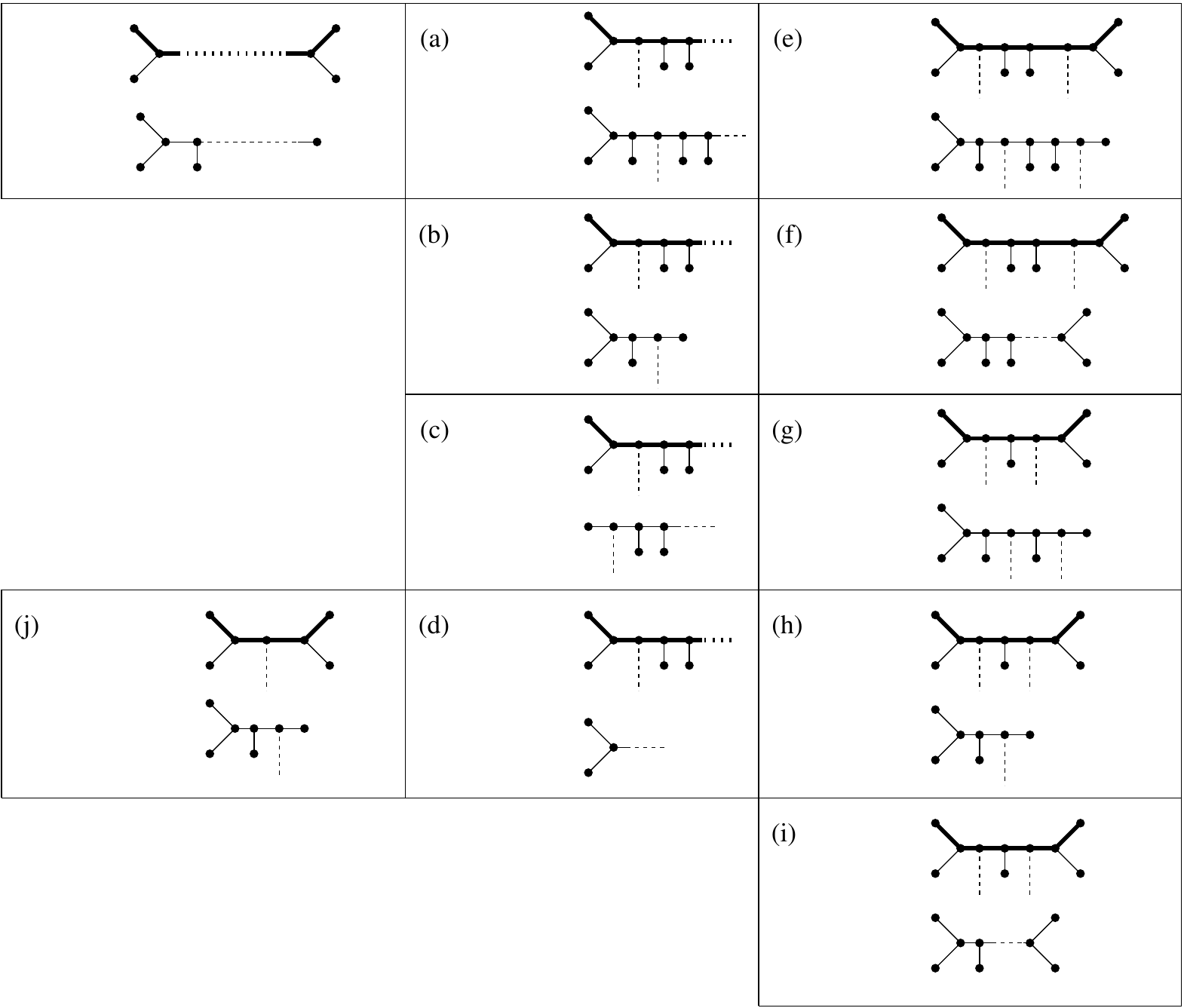}
\caption{Tree topologies checked by Algorithm 2. Path $P$, as used in Theorems \ref{thm:211sparse} and \ref{thm:211multiedge},
%and \ref{thm:211loop},
is indicated in bold. Solid lines are edges. Dotted and dashed lines are subtrees that can be optionally present in the tree. Figures (a)--(d) correspond to the situation when the $2|1|1$ side in the underlying generator is a simple edge, and (e)--(i) to when it is a multi-edge. Figure (j) corresponds to Theorem \ref{thm:211loop}, which deals with the situation when the side in the underlying generator is a loop.}
\label{fig:algorithm2}
\end{figure}

\end{document}